\documentclass[10pt,conference]{IEEEtran}

\let\labelindent\relax
\usepackage{enumitem}

\usepackage[english]{babel}
\usepackage[utf8]{inputenc}
\usepackage{amsthm} 
\usepackage{amsmath}
\usepackage{graphicx}
\usepackage{boxedminipage}
\usepackage{epsfig}
\usepackage{psfrag}
\usepackage{url}
\usepackage{comment}
\usepackage{afterpage}
\usepackage{pstricks}
\usepackage{graphics}

\usepackage{mdwtab}
\usepackage{wrapfig}
\usepackage{graphics}

\usepackage{algorithm}  
\usepackage{algpseudocode}

\usepackage{amssymb} 

\usepackage{makecell}

\usepackage{hhline}

\usepackage{tabularx}

\usepackage{epsfig}
\usepackage{fixfoot}
\usepackage{listings}
\usepackage{psfrag,wrapfig}
\usepackage{xspace,soul}
\usepackage[colorinlistoftodos]{todonotes}

\usepackage{subcaption}

\newcommand{\eg}{{\it e.g.,}\xspace}

\newcommand{\etal}{{\it et~al.}}
\newcommand{\ie}{{\it i.e.,}\xspace}\soulregister{\ie}{7}
\newcommand{\etc}{{\it etc.}}
\newcommand{\ci}{{\it (i) }}
\newcommand{\cii}{{\it (ii) }}
\newcommand{\ciii}{{\it (iii) }}

\newcommand{\ca}{{\it (a) }}
\newcommand{\cb}{{\it (b) }}
\newcommand{\cc}{{\it (c) }}
\newcommand{\cd}{{\it (d) }}
\newcommand{\ce}{{\it (e) }}
\newcommand{\cf}{{\it (f) }}
\newcommand{\cg}{{\it (g) }}
\newcommand{\ch}{{\it (h) }}
\newcommand{\cj}{{\it (j) }}

\usepackage{amsthm}
\theoremstyle{definition}
\newtheorem{definition}{Definition}
\newtheorem{example}{Example}
\newtheorem{theorem}{Theorem}
\newtheorem{observation}{Observation}

\renewcommand{\qedsymbol}{\hfill$\blacksquare$}

\newenvironment{tight_enumerate}{
\begin{enumerate}[leftmargin=18pt]
\vspace{-0.2\baselineskip}
  \setlength{\itemsep}{0.3\baselineskip}
  \setlength{\parskip}{0pt}
}{\vspace{-0.2\baselineskip}\end{enumerate}}

\newcommand{\ATTCK}{{\sffamily ScheduLeak}\xspace}
\soulregister{\ATTCK}{7} 
\newcommand{\ATTCKNF}{{ScheduLeak}\xspace}
\soulregister{\ATTCKNF}{7} 

\newcommand{\BUDGET}{\lambda}

\newcommand{\ARRVCOL}{\delta_v}

\title{A Novel Side-Channel in Real-Time Schedulers} 

\author{
\IEEEauthorblockN{Chien-Ying Chen\IEEEauthorrefmark{1}, Sibin Mohan\IEEEauthorrefmark{1}, Rodolfo Pellizzoni\IEEEauthorrefmark{2}, Rakesh B. Bobba\IEEEauthorrefmark{3} and Negar Kiyavash\IEEEauthorrefmark{4}} 
\IEEEauthorblockA{\IEEEauthorrefmark{1}Deptartment of Computer Science, University of Illinois at Urbana-Champaign, Urbana, IL, USA}
\IEEEauthorblockA{\IEEEauthorrefmark{2}Deptartment of Electrical and Computer Engineering, University of Waterloo, Ontario, Canada}
\IEEEauthorblockA{\IEEEauthorrefmark{3}School of Electrical Engineering and Computer Science, Oregon State University, Corvallis, OR, USA}
\IEEEauthorblockA{\IEEEauthorrefmark{4}Deptartment of Electrical and Computer Engineering, University of Illinois at Urbana-Champaign, Urbana, IL, USA}
Email: \{\IEEEauthorrefmark{1}cchen140,
\IEEEauthorrefmark{1}sibin, \IEEEauthorrefmark{4}kiyavash\}@illinois.edu,
\IEEEauthorrefmark{2}rodolfo.pellizzoni@uwaterloo.ca,
\IEEEauthorrefmark{3}rakesh.bobba@oregonstate.edu
}

\begin{document}

\maketitle

\thispagestyle{plain}
\pagestyle{plain}

\begin{abstract}
We demonstrate the presence of a novel scheduler side-channel in preemptive, fixed-priority real-time systems (RTS); examples of such systems can be found in automotive systems, avionic systems, power plants and industrial control systems among others. This side-channel can leak important timing information such as the future arrival times of real-time tasks.
This information can then be used to launch devastating
attacks, two of which are demonstrated here (on real hardware
platforms). Note that it is not easy to capture this timing information due to runtime variations in the schedules, the presence of multiple other tasks in the system and the typical constraints (\eg deadlines) in the design of RTS. Our \ATTCK algorithms demonstrate how to effectively exploit this side-channel. A complete implementation is presented on real
operating systems (in Real-time Linux and FreeRTOS). 
Timing information leaked by \ATTCKNF can significantly aid
other, more advanced, attacks in better accomplishing their goals.

\end{abstract}

\section{Introduction}
\label{sec::intro}

\begin{figure*}[t]
\centering
\includegraphics[width=0.93\textwidth]{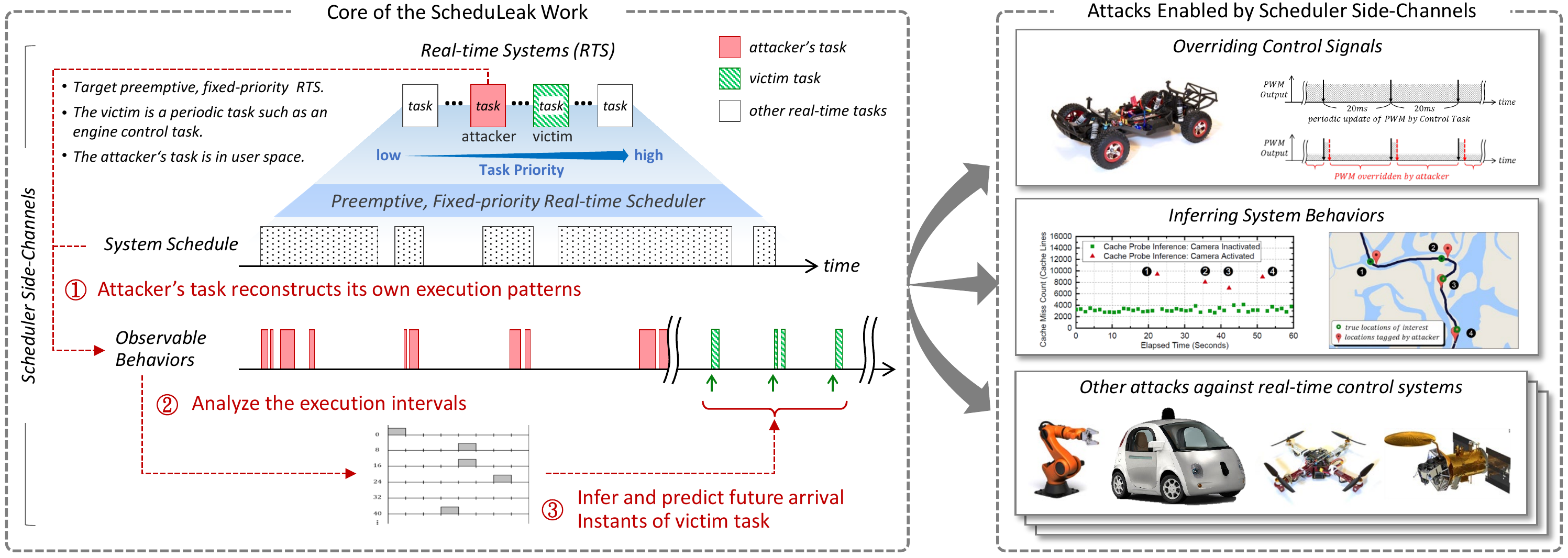}
\vspace{-0.3\baselineskip}
\caption{Overview of paper: We demonstrate how an unprivileged, low-priority task (in user space) can use the \ATTCKNF algorithms to infer execution behaviors of critical, high-priority periodic task(s). 
The extracted information is useful for helping other attacks achieve their primary goals (two such attack instances are implemented in this paper as possible use cases).
}
\label{fig:paper_overview}
\vspace{-1\baselineskip}
\end{figure*}


Consider the scenario where an adversary wants to attack an embedded real-time system (RTS) --
parts of autonomous cars, industrial robots, anti-lock braking systems in
modern cars, unmanned aerial vehicles (UAVs), power grid components, 
the NASA rovers, 
implanted medical devices, \textit{etc.} 
These systems
typically have limited memory and processing power, have very regimented
designs (stringent timing constraints for instance) and any unexpected actions
can be quickly thwarted.
Therefore, the opportunity to either steal a critical piece
of information or the ability to launch that attack which takes control of the system
can be very limited. As a consequence, attacks on such systems require significant system specific information. 
This ``information'' can take many forms -- from an
understanding of the design of the system, to knowledge of the critical
components (either software or hardware). The exact knowledge depends on the
type of attack and the target component. For example, say, \ca to steal important
information about when (and where) an on-board camera is used for
reconnaissance or \cb to take control away from the ground operator of a
remotely-controlled vehicle.

The one common underlying theme that pervades real-time systems (and something
that a would-be attacker should definitely address) is the importance of {\em
timing}.
``Timing'' includes:
\ci when certain events occur, \cii how often they occur, and, most
importantly for this paper, \ciii {\em when (and if) they will occur again in the
future}. In fact, a number of critical software components in real-time systems
are \textit{periodic} in nature. 
As we shall see, 
these periodic tasks 
present themselves as prime targets for attackers. 

So, how does one attack such systems, especially the periodic (and critical) 
components\footnote{We shall see potential end results of such attacks 
in Section \ref{sec:casestudies}.}? 
\newenvironment{myquotation}{\setlength{\leftmargini}{1.5em}\quotation}{\endquotation}
\begin{myquotation}
	\noindent\textit{We have discovered the presence of a scheduler-based side-channel 
	that leaks timing information in real-time systems -- in
	particular those with fixed priority tasks}. 
\end{myquotation}
The scheduler-based side-channel enables an \emph{unprivileged, low-priority task} to learn the precise timing behavior of the critical, periodic (victim) task(s) by simply observing its own execution intervals using a system timer.
This provides an attacker
with the ability to {\em infer the initial offset of the victim task and precisely predict its future arrival times at runtime}\footnote{In this paper, we do not focus on inferences of other task timing behaviors such as job start times or job completion times.}. We name the algorithms that exploit this
side-channel attack, ``\ATTCK''.

Figure \ref{fig:paper_overview} presents an overview of the side-channel and
also how the attacker can benefit from the scheduler side-channel-based information. 
The left side of the figure shows how a real-time system
consisting of fixed-priority tasks (the boxes at the top -- the victim is a periodic task while all other tasks can be either periodic or sporadic) that results in a
schedule (dotted boxes in the middle, with each task being indistinguishable from
the other at runtime) can be analyzed to extract the precise future arrival time points
(the green, upward arrows) of the victim task. The right-hand side of the
figure 
shows how this timing information of the critical task can be used to
launch other attacks that either leak more important information 
or destabilize the real-time control system. 
Note that without this precise timing information, an attacker is either
forced to guess when the victim task(s) will execute or launch the attacks at
random points in time -- both of which dilute the efficacy of the attack or
result in early termination of the system.

The extraction of this runtime timing information is non-trivial; 
main reasons include \ca the runtime schedule depends heavily on the state of the system at startup, initialization variables and environmental conditions and \cb real-time systems typically include multiple non-real-time tasks as well.
Even precise knowledge of \textit{all} statically-known system parameters is
insufficient to reconstruct the future arrival times of the victim. 
While a privileged attacker could target the scheduler of the system and extract the
requisite information, such access typically requires significant effort and/or
resources.  On the other hand, we are able to reconstruct the information with
the same level of precision using an \textit{unprivileged user space
application}.  
This is achieved by letting the attacker's application keep track of its own scheduling information.
Coupled with some easily obtainable information about the system (\textit{e.g.,} the victim task's period), the attacker can recreate the targeting timing information with high precision.

To be more specific, let's say that we want to override the (remote) control of
a rover. In many such systems, a periodic pulse width modulation (PWM) task
drives the steering and throttle. Without knowledge of 
when the PWM task is likely to update the motor control values,
the attacker is forced to employ brute force or random strategies to
override the PWM values. These could either end up being ineffective or lead to
the entire system being reset before the attack succeeds (see Section
\ref{sec:casestudies} for more details on this and another scenario). Armed
with knowledge from \ATTCK, our smart adversary can now override the PWM values
{\em right after} they have been written by the corresponding task --
effectively overriding the actuation commands.


Scheduler covert channels, where two processes covertly communicate using the
scheduler, have long been known (\eg~\cite{embeddedsecurity:volp2008,
ghassami2015capacity, embeddedsecurity:son2006}). In contrast, our focus is on
a \textit{side-channel that leaks execution timing behavior (not deliberately, as opposed to the scheduler covert channels)
of critical, high-priority real-time tasks to unprivileged, low-priority
tasks}. We focus on uniprocessor (\textit{i.e.,} single-core) systems with a preemptive,
fixed-priority real-time task scheduler~\cite{Buttazzo:1997:HRC,LiuLayland1973}
since they are the most common class of real-time systems deployed in practice
today~\cite{Liu:2000:RTS}. It is important for an attacker to \textit{stay
within the strict execution time budgets} allotted to the unprivileged task
-- especially during the phases
when it is trying to observe and reconstruct the victim's timing behavior.
Failing this requirement will likely cause other critical real-time
functionality to fail or trigger a watchdog timeout that resets the system,
leading to premature ejection of the attacker.  This
property is crucial during the `reconnaissance' phase of what has come to be known as
advanced persistent threat (APT) attacks~\cite{tankard2011advanced,
virvilis2013big}.  \textit{E.g.}, it has been reported that attackers had
penetrated and stayed resident undetected in the system for {\em months} before
they initiated the actual attack in the case of
Stuxnet~\cite{stuxnet:symantecdossier11}.  Once they had enough information
about system internals, they were able to craft effective
attacks tailored to that particular system.

The \ATTCKNF algorithms are implemented on both: \ca real hardware platforms running Real-Time Linux and \text{FreeRTOS} (for the two attack 
case studies) and \cb a simulator. We evaluate the
performance and scalability of \ATTCKNF in Section~\ref{sec::eval}, along with
a design space exploration (on the simulator). The results show that our methods are
effective at reconstructing schedule information and provide valuable
information for later attacks to better accomplish their attack goals.  To
summarize, the main contributions are: 
\begin{tight_enumerate}
\item Novel scheduler side-channel attack algorithms that can accurately
reconstruct the initial offsets and predict future arrival times of critical real-time
tasks in real-time systems (without
requiring privileged access) [Section~\ref{sec::approach}]

\item Analyses and metrics to measure the accuracy in predicting the execution
and timing properties of the victim tasks [Sections~\ref{sec::analysis}
and \ref{sec::eval}].

\item Implementation and case studies on real hardware platforms (\ie autonomous systems) running Real-time Linux and FreeRTOS [Section~{\ref{sec:casestudies}}].
\end{tight_enumerate}

\section{System and Adversary Model}
\label{sec:model}

\subsection{Time Model}
\label{subsec::time_model}

We assume that the attacker has access to a system timer on the target system and therefore time measured by the attacker has the resolution equal\footnote{Section~{\ref{sec:take_over_the_control}} demonstrates a case that the attacker may use a coarser time resolution and the proposed attack  algorithms would still work.} to this system timer.
The timer can be either a software or a hardware timer (\eg a 64-bit \emph{Global Timer} in FreeRTOS or a $\mathtt{CLOCK\_MONOTONIC}$-based timer in Linux).
We consider a discrete time model \cite{isovic2001handling}. 
We assume that a unit of time equals a timer tick (of the timer that the attacker can access) and
the tick count is an integer. 
All system and task parameters are
multiples of a time tick. We denote an interval starting from time point $a$
and ending at time point $b$ that has a length of $b-a$ by $[a,b)$ or $[a,b-1]$.



\subsection{System Model}
\label{sec::rts_model}
We consider a uniprocessor (\textit{i.e.,} single-core), fixed-priority, preemptive real-time system consisting of $n$ real-time tasks $\Gamma=\{\tau_1,...\tau_n\}$. A task can be either a periodic or a sporadic task.
Each task $\tau_i$ is characterized by
$(p_i, d_i, e_i, a_i, pri_i)$  where $p_i$ is the period (or the minimum
inter-arrival time), $d_i$ is the relative deadline,
$e_i$ is the worst-case execution time (WCET),
$a_i$ is the initial task offset (\ie the arrival time) and
$pri_i$ is the priority.  We assume that every task has a distinct period\footnote{This assumption is in line with existing standards in the design of real-time tasks to ensure 
distinct periods/priorities. For example, AUTOSAR (a standardized automotive software architecture) tools map runnables/functions activated by the same period to a single task to reduce context switch/preemption overheads.} and
that a task's deadline is equal to its period~\cite{LiuLayland1973}.
We use the same symbol $\tau_i$ to represent a task's job (or instance) for simplicity of notation.
We assume that task release jitter is negligible. Thus, any two adjacent arrivals of a periodic task $\tau_i$ has a constant distance $p_i$.
%
We further assume that each task is assigned a distinct priority and that the taskset is schedulable by a fixed-priority, preemptive real-time scheduler.
Let $hp(\tau_i)$ denote the set of tasks that have
higher priorities than that of $\tau_i$ and $lp(\tau_i)$ denote the set of tasks that
have lower priorities than $\tau_i$.
We define an ``execution interval'' of a task to be 
an interval of time $[a,b)$ during which the task runs continuously.
If $\tau_i$ is preempted then the execution
will be partitioned into multiple execution intervals, each of which has length less than $e_i$.

\begin{table}[t]\footnotesize
\centering
\caption{A summary of the system and adversary model.}
\vspace{-0.5\baselineskip}
\label{tab:assumptions}
\begin{tabular}{|l||l|}
\hline
\multicolumn{2}{|c|}{\textbf{Real-Time System Assumptions}}                                      \\ \hline 
A1 & A preemptive, fixed-priority real-time scheduler is used. \\ \hline 
A2 & The victim task is a periodic task. \\ \hline\hline
\multicolumn{2}{|c|}{\textbf{Attacker's Capabilities (Requirements)}}                                 \\ \hline
R1 & The attacker has the control of one user-space task (observer \\ & task) that has a lower priority than the victim task. 
\\ \hline
R2 & The attacker has knowledge of the victim task's period.                                                   \\ \hline
R3 & The attacker has access to a system timer on the system.
	\\ \hline \hline
\multicolumn{2}{|c|}{\textbf{Attacker's Goals}}                                      \\ \hline
G1 & Infer the victim task's initial offset and predict future arrivals.
	\\ \hline
\end{tabular}
\vspace{-1\baselineskip}
\end{table}


\subsection{Adversary Model}
\label{sec::adver_model}
We assume that an attacker is interested in targeting {\em one of the 
critical tasks in the system} that we henceforth refer to as a ``\emph{victim task}'', denoted by $\tau_v \in \Gamma$.
We also assume that $\tau_v$ 
is a \textit{real-time, periodic}
task. 
Many critical functions in real-time control systems are periodic
in nature, 
\eg the code that controlled the frequency of the slave variable-frequency
drives in the Stuxnet example~\cite{stuxnet:symantecdossier11}. 
In all such cases, the period of the task is strictly related to the characteristics of the physical system and thereby can be deduced from the physical properties; hence, we can assume that the attacker is able to gain knowledge of the victim task's period beforehand. 
It is common that, before attacking complex systems (\textit{e.g.,} CPS), adversaries will study the design and details of such systems.
However, the attacker does not know the initial conditions at system start-up 
({\eg} the task's initial offset) and may not have information on all the tasks in the system. 
All other tasks in the system can be either periodic, sporadic or non-real-time, depending on the design of the system.
Hence, the methods developed in this paper can target systems that have a mix
of periodic, sporadic and non-real-time tasks.

The ultimate goal varies with adversaries and the systems
under attack. For example, in advanced persistent threat (APT)  attacks~\cite{tankard2011advanced,
virvilis2013big}, one may plan to interfere with the operations of critical
tasks, eavesdrop upon certain information via shared resources
or even carry out debilitating attacks at a critical juncture when the victim system is most vulnerable.
Oftentimes, such attacks require the attacker to precisely gauge the timing
properties of victim tasks. 
In this paper, we introduce attack algorithms that help an attacker obtain this valuable information during the reconnaissance stage.
In this context, the main goal of the attacker is to {\em precisely infer 
when the victim task is scheduled to run}
in the near future (\textit{i.e., the future arrival times}). 

Note that our focus in this paper is on how to 
reconstruct 
the timing behavior of a higher-priority periodic victim task using the scheduler side-channel without violating the real-time constraints. 
We do this from the vantage point of a compromised, lower-priority (``observer'') task.
We do not focus on {\em how} attackers get access to the observer task. They could use any number of known methods -- from compromised insiders, to supply chain vulnerabilities in a multi-vendor development
model (as is usually practiced for the design and development of large, complex systems such as 
aircraft, automobiles, industrial control systems, \etc)~\cite{embeddedsecurity:mohan2015}, to vulnerabilities in the software and network among others. 
Recent work has demonstrated that real-time systems like commercial drones contain design flaws and hence are vulnerable to compromise~\cite{pleban2014hacking,samland2012ar}. 
%
The details of gaining access to an observer task are out of scope for this paper. Nevertheless, it is important to note that we {\em do not require the observer task to be a privileged task in the system}.
A summary of assumptions, attacker's capabilities and goals is given in Table~\ref{tab:assumptions}.

\subsection{Observer Task}
\label{sec::observer_task}
As previously mentioned, we refer to the lower-priority task  that the attacker controls as an ``{\em observer
task}'' and it is denoted by $\tau_o \in \Gamma$. 
It can be a user-space task.
The only constraint we place on $\tau_o$ is that it has a lower priority than the victim task, $pri_v > pri_o$. 
The observer task can be either a periodic or a sporadic task and its period (or its minimum inter-arrival time) can be shorter or longer than the victim task.
In particular, being a periodic task is a more restrictive condition since it reduces the flexibility available to an attacker (this will be clearer as we introduce the algorithms).
That is, the case where a periodic observer task with a period $p_o$ and priority $pri_o$ can succeed, a sporadic observer task (by picking the same $p_o$ as the minimum inter-arrival time and the same priority $pri_o$) can also succeed.
Therefore, when analyzing the attack capabilities in Section~\ref{sec::analysis}, we will consider a periodic observer task (or a sporadic observer task running at a constant inter-arrival time). 

In this paper, {\em we use the observer task to infer the initial offset $a_v$ that can
be used to predict future arrivals of the victim task}. 
We let the observer task ``monitor'' its own execution  intervals by using a system timer.
Note that reading system time does not require privilege in most operating systems ({\eg} invoking $\mathtt{clock\_gettime()}$ in Linux).
The key idea here is that the intervals when the observer task is active 
{\em cannot contain the victim task's} execution or its arrival time point since the victim would have preempted the observer task. 
However, there are also other higher-priority tasks that can impact the observer task's execution behaviors.
To the attacker, the challenge is to then filter out unnecessary information 
and extract the correct information about the victim task.
This is explained in the following section.

\begin{figure*}
\centering
\begin{minipage}{.36\textwidth}
  \centering
  \begin{subfigure}[t]{0.99\textwidth}
        \centering
        \includegraphics[width=0.99\columnwidth]{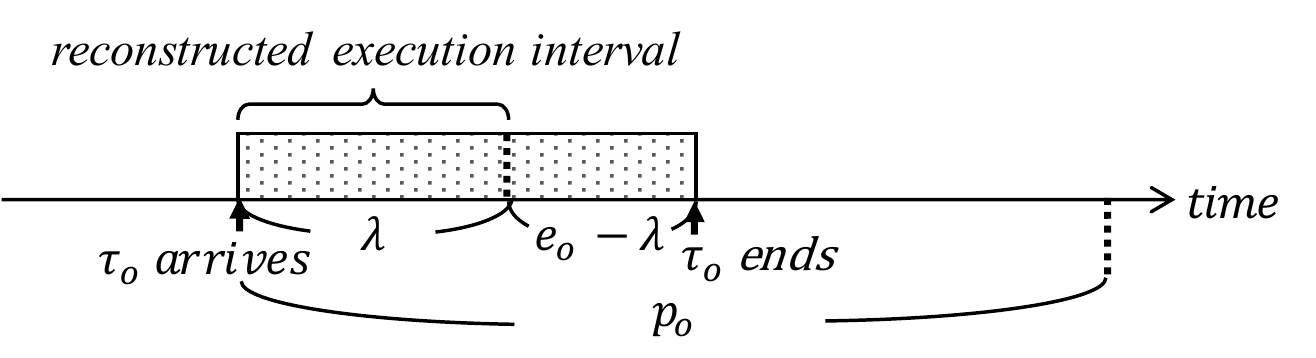}
        \vspace{-1.5\baselineskip}
        \caption{No preemption has occurred.}
    \end{subfigure}%
    
    \begin{subfigure}[t]{0.99\textwidth}
        \centering
        \includegraphics[width=0.99\columnwidth]{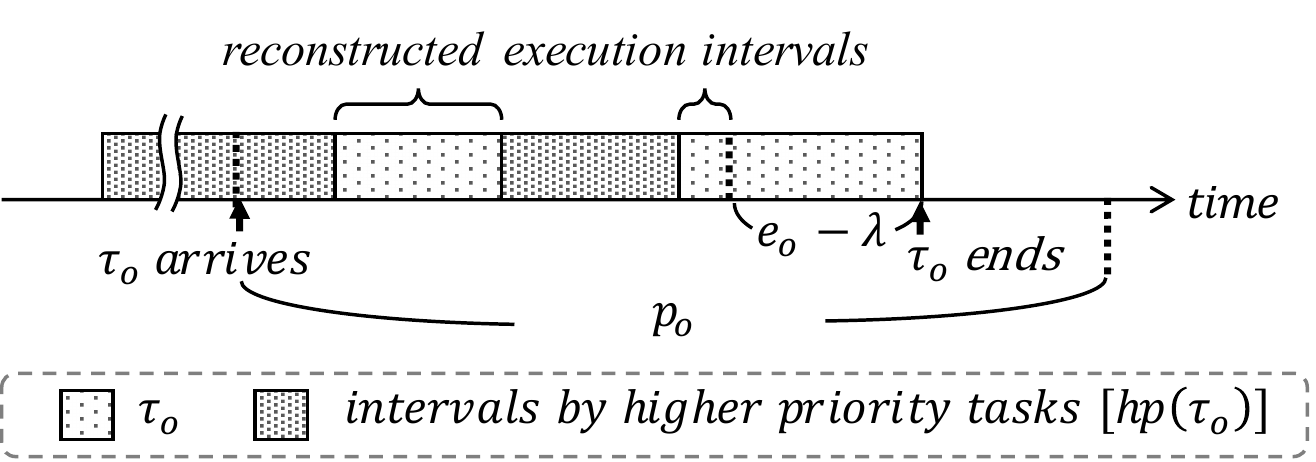}
        \vspace{-1.3\baselineskip}
        \caption{Some tasks $\tau_i \in hp(\tau_o)$ preempt $\tau_o$.}
    \end{subfigure}
    \vspace{-0.3\baselineskip}
    \caption{Examples of reconstructed execution intervals of the observer task. The total length of the	reconstructed execution interval(s) is $\BUDGET$ that leaves $e_o - \BUDGET$ for $\tau_o$ to perform original task functions.}
    \label{fig:construct_ei}
\end{minipage}%
\hfill
\begin{minipage}{.3\textwidth}
  \centering
  \includegraphics[width=0.99\columnwidth]{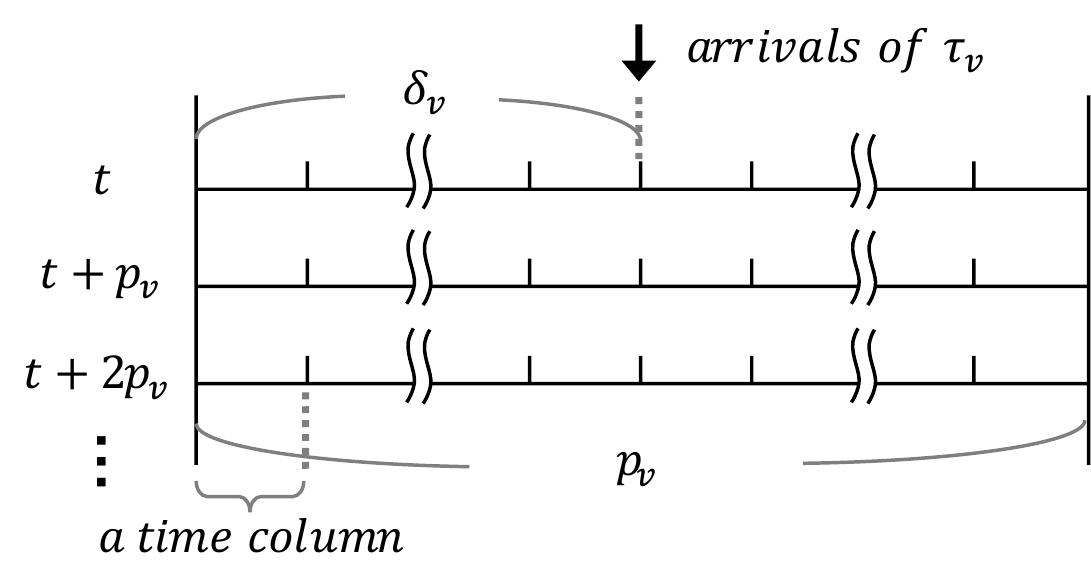}
\vspace{-1.3\baselineskip}
\caption{The skeleton of a schedule ladder diagram. 
    The start time $t$ of the diagram 
    (\ie the beginning of the top timeline) 
    is an arbitrary point in time, assigned by the attacker. 
    The width of each timeline matches the victim task's period $p_v$. 
    The relative offset between the start time $t$ and the true arrival column is defined by $\ARRVCOL$.}
\label{fig:ladder_example}
\end{minipage}
\hfill
\begin{minipage}{.3\textwidth}
  \centering
  \includegraphics[width=0.99\columnwidth]{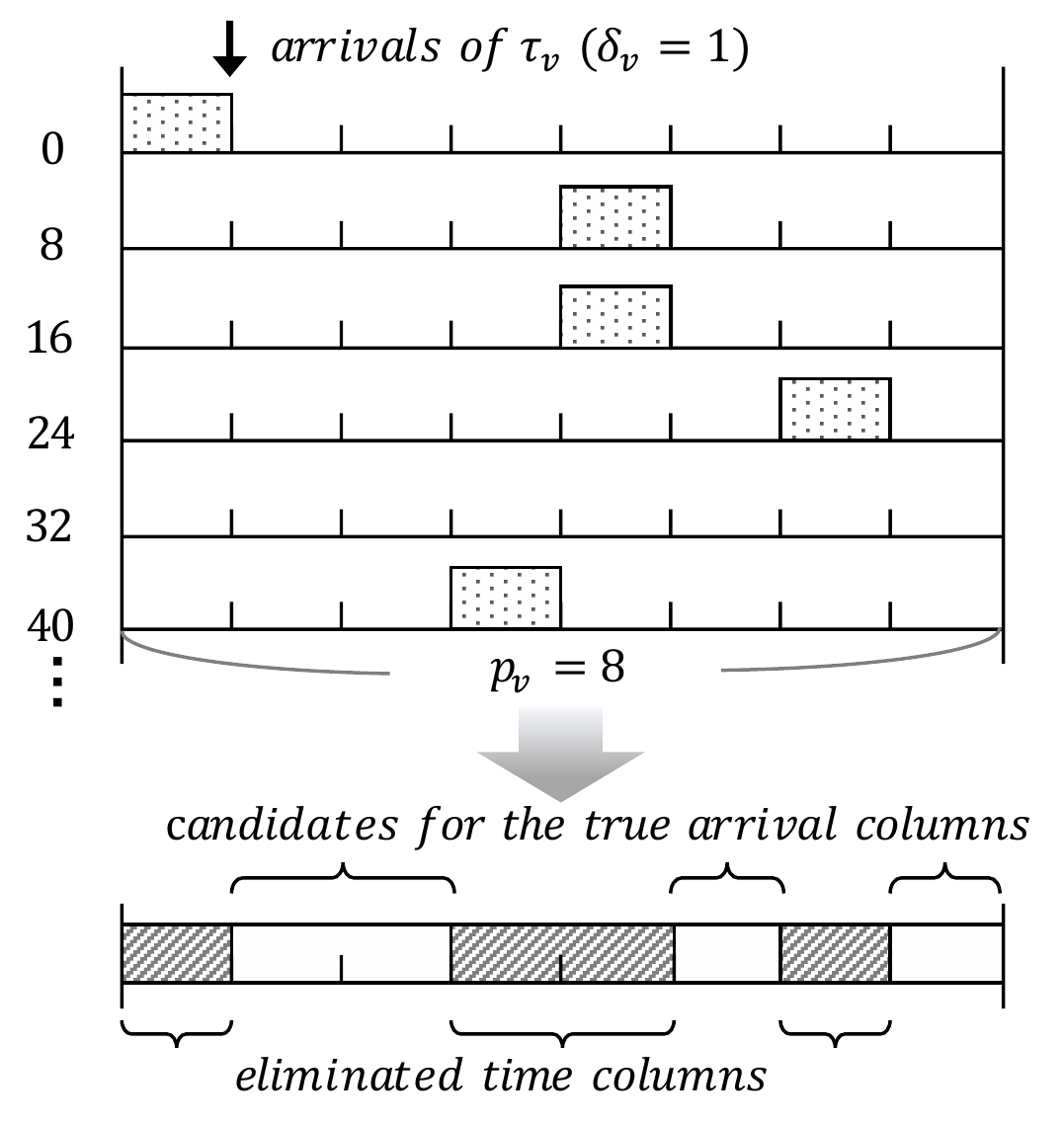}
\vspace{-1.3\baselineskip}
\caption{The processed schedule ladder diagram for Example~\ref{example:analyze_ei}.}
    \label{fig:ei_and_union_example}
\end{minipage}
\vspace{-1\baselineskip}
\end{figure*}

\section{\ATTCKNF}
\label{sec::approach}

\subsection{Overview}
\label{subsec::overview}

We now introduce the core algorithms. 
The main idea is that the victim task cannot run while the observer task is running since the latter has a lower priority. By reconstructing the observer task's own execution intervals and analyzing those intervals based on the victim task's period, we may infer the \textit{initial offset} and \textit{future arrival times} for the victim task.
A high-level overview of the various analyses stages in our proposed \ATTCKNF algorithms includes: 

\begin{tight_enumerate}
\item {\it Reconstruct execution intervals of the observer task}: 
first, the observer task uses a system timer to measure and reconstruct
{\em its own} execution intervals (\ie times when it itself is active). [Section~\ref{sec::reconstruct_ei}]

\item {\it Analyze the execution intervals}: 
The reconstructed execution intervals are organized in a ``{\em schedule ladder diagram}'' -- a timeline that is divided into windows that match the period of the victim task. [Section~\ref{sec::analyze_ei}]

\item {\it Infer the victim task's initial offset and future arrivals}: 
in the final step, 
the initial offset for the victim task is inferred.
This information is then used to predict the future arrivals of the victim. 
Since the victim task is periodic in nature, the offset from the start of its own window translates to the offset from startup when the first instance of the victim task executed.
[Section~\ref{sec::infer_initial_offset}]
\end{tight_enumerate}

\subsection{Reconstruction of Execution Intervals}
\label{sec::reconstruct_ei}

The first step is to reconstruct the observer task's
execution intervals.
We implement a function in the observer task that keeps track of time read from the system timer.
By examining the polled time stamps, preemptions (if any) can be identified and the execution intervals of the observer task can be reconstructed.
While this function seems straightforward, ensuring that it respects real-time constraints ({\ie} all real-time tasks must meet their deadlines) is critical.
That is, the observer task should not run more than its WCET, $e_o$.
Furthermore, even if the attacker does not exceed the
allocated execution budget for itself, it may want to save some budget for other purposes such as performing the analyses to 
reconstruct the timing information of the victim. 
Hence, we define a parameter,
$\BUDGET$, whose value is set by the attacker, to limit the running time of the aforementioned function for the observer task in each period. This ``maximum reconstruction time'', $\BUDGET$, 
is an integer in the range $0 \leq \BUDGET \leq e_o$. The total length of the reconstructed execution intervals is $\BUDGET$ in each period and this leaves the timespan $e_o-\BUDGET$ for the observer task 
to carry out other computations. 
As a result, the service levels guaranteed  by the original (clean) system is still maintained -- thus reducing the risk of triggering system errors. 
On the flip side, the 
attacker may
not be able to capture all possible execution intervals and this could reduce the fidelity/precision of the final results. 
Section~\ref{sec::sigma} discusses how to compute good values for $\BUDGET$. Figure~\ref{fig:construct_ei} shows examples of reconstructed execution intervals.
The function for reconstructing an execution interval of the observer task while taking {$\BUDGET$} into account is detailed in Appendix{\ref{sec::reconstruct_an_execution_interval_algo}} as Algorithm {\ref{algo:construct_ei}}.

\subsection{Analysis of Execution Intervals}
\label{sec::analyze_ei}

Once the observer task's execution intervals are reconstructed, we 
analyze the data to extract information about the victim task.
We organize the observer's execution intervals into a timeline split
into lengths of the victim task's period $p_v$
(recall that $p_v$ is one of the known quantities for the attacker).
The purpose of this step is to place
the execution intervals of the observer task within periodic windows of the victim task.  
The timeline split into windows of length that matches the victim task's period allows the
attacker to see how the observer task's execution intervals are influenced by
the victim task as well as other higher-priority tasks.

To better illustrate the idea of the timeline and the proposed algorithms, we
will use a ``\emph{schedule ladder diagram}'' (defined below) to represent the
construction of the timeline in this paper. The rows in the schedule ladder diagram can
be merged into a single-line timeline (and is just an analytical ``trick'').
A schedule ladder diagram is a skeleton consisting of a set of
adjacent timelines of equal lengths -- that match the victim task's period
$p_v$. 
The start time of the top section can be an 
arbitrary point in time assigned by the attacker (\eg the time instant when the 
algorithms are first invoked).
The columns in the schedule ladder diagram are ``unit time columns''. So, there are $p_v$ time columns. 
That is, the schedule ladder diagram has the same time resolution as the reconstructed execution intervals.
The skeleton of a schedule ladder diagram is illustrated in Figure~\ref{fig:ladder_example}.
From the diagram, plotted based on $p_v$, we make the
following observation:
\begin{observation}
\label{observation:periodic_arrival}
Any \emph{schedule ladder diagram} of $\tau_v$ must contain
exactly one arrival instance of $\tau_v$ in every row.
All arrivals of $\tau_v$ are located in the same time column.
\end{observation}

This observation is true because $\tau_v$ is a periodic task that arrives every $p_v$ time units and the schedule ladder diagram is plotted with
its interval equal to $p_v$. We define the column where the arrivals of
the victim task are located as the ``true arrival column'', denoted by $\ARRVCOL$.  
Thus, the correlation between the initial offset $a_v$ and the true arrival column $\ARRVCOL$ can be derived by $(t + \ARRVCOL - a_v) \ mod \ p_v = 0$, where $t$ represents the (arbitrary) start time of the schedule ladder diagram assigned by the attacker.
This is also depicted in Figure~\ref{fig:ladder_example}. 
Based on this observation, we define the following theorem with respect to the observer task's executions on the schedule ladder diagram:


\begin{theorem}
\label{th:blank_columns}
The observer task's execution intervals do not appear at the time columns $[\ARRVCOL, \ARRVCOL+bcet_v)$, where $bcet_v$ is the best case execution time of $\tau_v$.
\end{theorem}

\begin{proof}
From Observation~\ref{observation:periodic_arrival}, the victim task $\tau_v$ arrives regularly at time column $\ARRVCOL$.
If there exists lower priority tasks $lp(\tau_v)$ in execution at $\ARRVCOL$ column, the victim task preempts such tasks until it finishes its job with length of $bect_v$ at a minimum. In the case that there exists higher priority tasks $hp(\tau_v)$ that are executing or arriving during $[\ARRVCOL, \ARRVCOL+bcet_v)$, the victim task $\tau_v$ is preempted. Under this circumstance, if the observer task $\tau_o$ had arrived during $[\ARRVCOL, \ARRVCOL+bcet_v)$, as a lower priority task, it is also preempted.
	Therefore, the time columns $[\ARRVCOL, \ARRVCOL+bcet_v)$ {\em cannot} contain the
	execution intervals of the observer task.
\end{proof}
\vspace{-0.3\baselineskip}

In other words, the time columns where the observer task $\tau_o$ can ever appear are not the true arrival column $\ARRVCOL$. 
To this end, it's easier to think of the problem as the process of eliminating those such time columns. 
If we place the obtained execution intervals of $\tau_o$ on the schedule ladder diagram and remove the corresponding time columns, then, there must exist at least an interval of continuous time columns, of which the length is equal to or greater than $bcet_v$, that is not removed in the end. 
Those time columns are candidates for the true arrival time of $\tau_v$. 
There may also exist time columns
that are not removed 
due to other higher-priority tasks.
Yet, since other tasks have distinct arrival periods (or random arrivals for sporadic tasks), those time columns tend to be scattered (compared to $[\ARRVCOL, \ARRVCOL+bcet_v)$) and are expected to be eliminated as more execution intervals of the observer task are collected. 
In practice, our results indicate that this process works effectively and is mostly stabilized after an attack duration of $5 \cdot LCM(p_o,p_v)$ (see Section~\ref{sec::exp_attack_duration}).

\begin{example}
\label{example:analyze_ei}
Consider an RTS consisting of four tasks $\Gamma=\{\tau_1, \tau_o,	\tau_v, \tau_4\}$.
For the sake of simplicity, we assume that all tasks are periodic in this example (though our analysis can work with periodic, sporadic and mixed systems as well). 
The task parameters are presented in the table below (on the left).
	Note that $pri_i > pri_j$ means that $\tau_i$ has a higher number 
	than $\tau_j$. Thus, task $\tau_1$ has the lowest priority while 
	task $\tau_4$ has the highest priority and $\tau_v$ has higher priority 
	than $\tau_o$. Let the maximum reconstruction duration $\BUDGET$ be $1$ 
	and the start time of the attack be $0$ (as a result, $a_v$ equals $\ARRVCOL$ in this example).  Assuming the attacker has 
	executed the first step/algorithm for some duration, the table 
	below lists the reconstructed execution intervals 
	of the observer task.

\begin{center}\footnotesize
\vspace{-0.5\baselineskip}
\begin{tabular}{|c||c|c|c|c|}
\hline 
 & $p_i$ & $e_i$ & $a_i$ & $pri_i$\\ 
\hline \hline 
$\tau_1$ & 15 & 1 & 3 & 1 \\ \hline 
$\tau_o$ & 10 & 2 & 0 & 2 \\ \hline 
$\tau_v$ & 8 & 2 & 1 & 3\\ \hline
$\tau_4$ & 6 & 1 & 4 & 4 \\ \hline 
\end{tabular} 
\quad
\begin{tabular}{|c|}
\hline 
Reconstructed \\ 
Execution Intervals \\ 
\hline \hline 
[0,1) \\ 
\hline 
[12,13) \\ 
\hline 
[20,21) \\ 
\hline 
[30,31) \\ 
\hline 
[43,44) \\ 
\hline 
\end{tabular} 
\end{center}

Note that since $\tau_1$ has priority lower than the observer task $\tau_o$, it
does not influence the execution of $\tau_o$. Then, we place
the reconstructed execution intervals in a schedule ladder diagram of
width equal to the victim task's period $p_v$. This operation is shown in Figure~\ref{fig:ei_and_union_example}. 
To better understand the effectiveness of the schedule ladder diagram in profiling the victim task's behavior, we plot the original, complete, schedule  on the ladder diagram in Figure~\ref{fig:schedule_on_ladder_example} in Appendix so that readers get a better sense of it. 
This gives us an insight into the relation between the execution intervals of $\tau_o$ and that of the victim task.

From the schedule ladder diagram in Figure~\ref{fig:ei_and_union_example}, 
we remove the time columns that are occupied by the observed execution intervals.
The results are shown at the bottom of Figure \ref{fig:ei_and_union_example}. 
What's left are candidate time columns that contain the true arrival times for the victim that we want to extract. These intervals are passed to the final step to infer the initial offset/arrival times
of the victim task.\qedsymbol

\end{example}
\vspace{-0.5\baselineskip}

\subsection{Inference of Initial Offset and Future Arrival Instants}
\label{sec::infer_initial_offset}

We now get to the final step -- inferring the future arrival instants of
the victim task -- our original objective. But, first, we need to calculate
the initial offset of the victim task. 
%
What we get from the previous step 
is a set of intervals of candidate time columns that
contains the true arrival column of the victim task. The number of intervals depends on the
number of collected execution intervals as well as the ``noise'' introduced by 
other, higher-priority, tasks (hence, there is no guarantee that all false time columns 
can be eliminated in the end). 
However, as observed from our experiments and based on Theorem~\ref{th:blank_columns}, 
the false time columns tend to
be scattered.
Therefore, we take the largest interval as 
our inference that may contain the true arrival column of the victim task.
We then pick the start of this interval as the inferred true arrival column, denoted by $\hat{\ARRVCOL}$.
While this strategy is not always guaranteed to
succeed, our evaluation (both case studies in Section~\ref{sec:casestudies} and performance evaluation in Section~\ref{sec::eval}) shows that we are able to 
achieve a high degree of precision for the inference.
The required initial offset, denoted by $\hat{a_v}$, can then be derived as $\hat{a_v}=(t + \hat{\ARRVCOL}) \ mod \ p_v$, where $t$ represents the start time of the schedule ladder diagram.

\begin{example}
\label{example:infer}
The intervals obtained from Example~\ref{example:analyze_ei}
	correspond to the time columns $[1,3), [5,6)$ and $[7,8)$. According
	to the algorithm, the largest interval, $[1,3)$, is selected. The
	starting point of such an interval is then taken as the inference of
	the victim task's true arrival column, which becomes $\hat{\ARRVCOL}=1$. In this example,
	the true arrival column is $\ARRVCOL=1$. Therefore, the algorithms
	correctly infer the true arrival column of the victim task and the initial offset can be derived accordingly.\qedsymbol
\end{example}

Now, the future arrivals of the victim task can easily be computed by
$\hat{a_v} + p_v \cdot T$, $T \in \mathbb{N}$, 
where $\hat{a_v}$ is the inferred initial offset of $\tau_v$, $p_v$ is the period of $\tau_v$ and $T$ is the desired arrival number. The result of this calculation is the {\em exact time of the $T^{th}$ arrival of the victim task}.


\begin{figure*}
	\centering
	\includegraphics[width=0.85\textwidth]{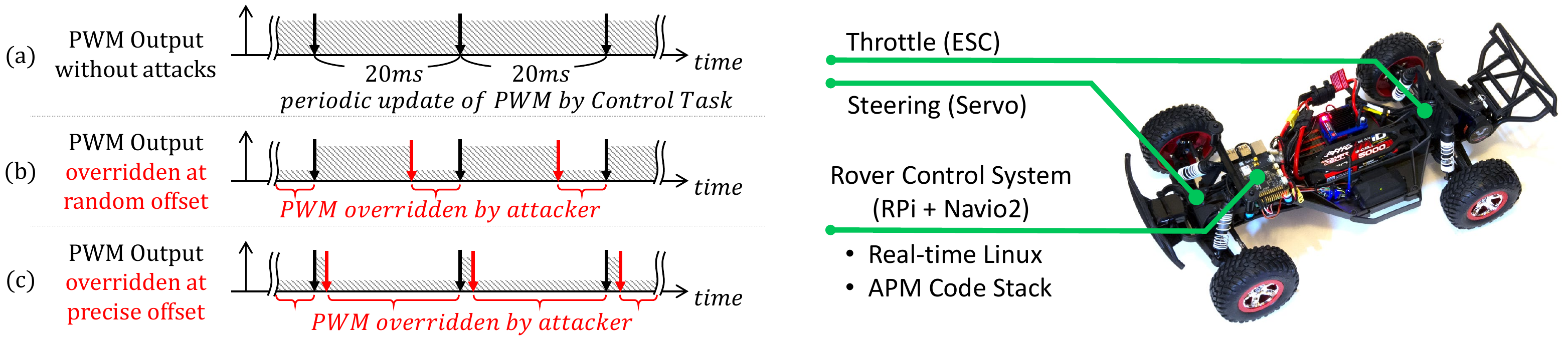}
	\vspace{-0.4\baselineskip}
    \caption{An illustration of PWM channels on a rover system. (a) The PWM outputs are updated periodically by a $50Hz$ task. (b) A naive attack issuing the PWM updates at random instants may not be effective. (c) By carefully issuing the PWM updates right after the original updates, the PWM outputs can be overridden.}
    \label{fig:exp_pwm_override}
    \vspace{-1\baselineskip}
\end{figure*}
\section{Analysis of Algorithms}
\label{sec::analysis}

\subsection{Analyzing Attack Capability}
\label{sec::determine_capability}

We now discuss how to determine the attack capability or effectiveness of the observer task with respect to the victim task.
%
That is, in this context, whether the observer task can remove all false time columns, and hence, correctly infer the arrival information of the victim task.
Note that the analysis presented in this section focuses on the observer task being a periodic task since, as we mentioned in Section~{\ref{sec::observer_task}}, it is a more restrictive condition to an attacker.
Given the same target system, a sporadic observer task may perform better 
as the sporadic task naturally has more flexible arrivals that are constrained only by its minimum inter-arrival time.
%
%


A conservative condition ensuring that all false time columns can be removed from the schedule ladder diagram of {$\tau_v$} is: when the observer task's execution intervals appear in all possible time columns.
Therefore, we first analyze how the observer task's execution relates to the victim task's execution.
When considering both $\tau_v$ and $\tau_o$ as periodic tasks, 
we have the following observation and theorem: 
\begin{observation}
\label{observation:lcm}
In the schedule ladder diagram, the \emph{offset} between the time column of each observer task's arrival (\ie the scheduled execution) and the true arrival column repeats after their {\emph{least common multiple}}, $LCM(p_o, p_v)$.
\qedsymbol
\end{observation}

\begin{theorem}
\label{th:full_coverage}
If the given observer task $\tau_o$ and the victim task $\tau_v$ satisfy the inequality
$e_o \geq GCD(p_o, p_v)$,
then the scheduled execution of $\tau_o$ is guaranteed to appear in all time columns of the schedule ladder diagram of $\tau_v$.
\end{theorem}

\label{pf:full_coverage}
\begin{proof}
From Observation~\ref{observation:lcm}, 
the time column offset of the observer task's execution repeats every $LCM(p_o, p_v)$.
Therefore, the aforementioned condition (\ie the scheduled execution of $\tau_o$ appears in all possible time columns) can be described by the 
inequality
$\frac{LCM(p_o,p_v)}{p_o} \cdot e_o \geq p_v$.
%
Then, by using $LCM(p_o,p_v)=\frac{p_o p_v}{GCD(p_o,p_v)}$, we can derive a condition for $e_o$ that guarantees that the observer task can detect the arrivals of the victim task to be $e_o \geq GCD(p_o, p_v)$.
\end{proof}

From Theorem~\ref{th:full_coverage}, we find that the observer task's scheduled execution can appear in some of the time columns more than once during $LCM(p_o,p_v)$ when $e_o>GCD(p_o,p_v)$.
The redundant coverage means that the false time columns will be visited by $\tau_o$ more frequently when compared to the lower ratio of $e_o$ to $GCD(p_o,p_v)$.
In contrast, if $e_o<GCD(p_o,p_v)$, then not all the false time columns can be covered and examined by the observer task.
To better profile the observer task's coverage, we further define a \emph{coverage ratio} that depicts the observer task's capability against the victim task
as follows
\begin{definition}
\label{def:coverage_ratio}
(Coverage Ratio)
The \emph{coverage ratio}, denoted by $\mathbb{C}(\tau_o, \tau_v)$, is computed by 
\vspace{-0.5\baselineskip}
\begin{equation}
\label{equation:coverage_ratio}
\mathbb{C}(\tau_o, \tau_v) = \frac{e_o}{GCD(p_o,p_v)}
\end{equation}
\end{definition}

The \emph{coverage ratio} can be loosely interpreted as the proportion of the time columns 
where the observer task can potentially appear
in the \emph{schedule ladder diagram}.
If all $p_v$ time columns 
can be covered by the observer task,
then $\mathbb{C}(\tau_o,\tau_v) \geq 1$.
Otherwise $0 \leq \mathbb{C}(\tau_o,\tau_v) < 1$.

\subsection{Choosing The Maximum Reconstruction Duration $\BUDGET$}
\label{sec::sigma}
Recall that, the maximum reconstruction duration $\BUDGET$ is used to limit the amount of execution time (in a period) taken up by the observer task for running the attack algorithms.
As the attacker wants to stay stealthy and minimize disruption to the original functionality, it is desirable to use a $\BUDGET$ value as small as possible. 
The remaining execution time $e_o-\BUDGET$ can then be used by the attacker to deliver the original functionality of $\tau_o$ while making progress on the capturing of execution data.
Based on this idea, $\BUDGET$ can be determined by:
\begin{equation}
\BUDGET = 
\begin{cases}
    GCD(p_o,p_v) & \text{if } \mathbb{C}(\tau_o,\tau_v) \geq 1\\
    e_o          & \text{otherwise}
\end{cases}
\end{equation}

In the case of  $\mathbb{C}(\tau_o,\tau_v) \geq 1$, the observer task has redundant coverage. Since a one-time coverage is sufficient for the observer task to examine all $p_v$ time columns, the additional coverage can be traded for other purposes.
Otherwise ($\mathbb{C}(\tau_o,\tau_v) < 1$), the attacker may need to utilize all its computational resource for the attack.

\section{Evaluation Metrics}
\label{sec::eval_metrics}

\noindent
To evaluate 
\ATTCKNF,
we define the following two metrics: 

\noindent
\textbf{(i) Inference Success Rate:}
We define an inference 
to be successful if attacker is
able to {\em exactly} infer the victim task's initial offset (recall from Section
\ref{sec::infer_initial_offset} that once we know the initial offset, we can 
easily predict the future arrival instants). Therefore, the result of
an inference is either true or false. The inference success rate is an {\em average
of the true/false results} for a given test condition for a set of task sets.

\noindent
\textbf{(ii) Inference Precision Ratio:}
In the case that the inference is not exact, we define a metric to evaluate the {\em degree}
of the inference precision (\ie how close we got to the actual values). In this paper, 
the inference target is the initial offset of the victim task. We first compute the distance between the inference
and the true value by
$\epsilon = \left | \hat{a_v} - a_v \right |$
, where $a_v$ is the 
initial offset of the victim task and $\hat{a_v}$ is the inferred initial offset. 
We then define the inference precision ratio:
\begin{definition}
(Inference Precision Ratio)
The inference precision ratio, denoted by $\mathbb{I}_v^o$, is computed by
\begin{equation}
\mathbb{I}_v^o = 
\begin{cases}
    1 - \frac{p_v - \epsilon}{\frac{p_v}{2}} & \text{if } \epsilon > \frac{p_v}{2}\\
    1 - \frac{\epsilon}{\frac{p_v}{2}}       & \text{otherwise}
\end{cases}
\end{equation}
\end{definition}
\vspace{-0.5\baselineskip}

\noindent
The inference precision ratio is a real number within $0 \leq
\mathbb{I}_v^o \leq 1$.  
It allows us to know how close the inference is to the true initial offset.
$\mathbb{I}_v^o = 1$ indicates that the inference of
the initial offset $a_v$ is absolutely correct.



\section{Evaluation Using Case Studies \\on Real Platforms}
\label{sec:casestudies}

\begin{figure*}
  \centering
  \begin{subfigure}{0.36\textwidth}
		\centering
        \includegraphics[width=0.98\columnwidth]{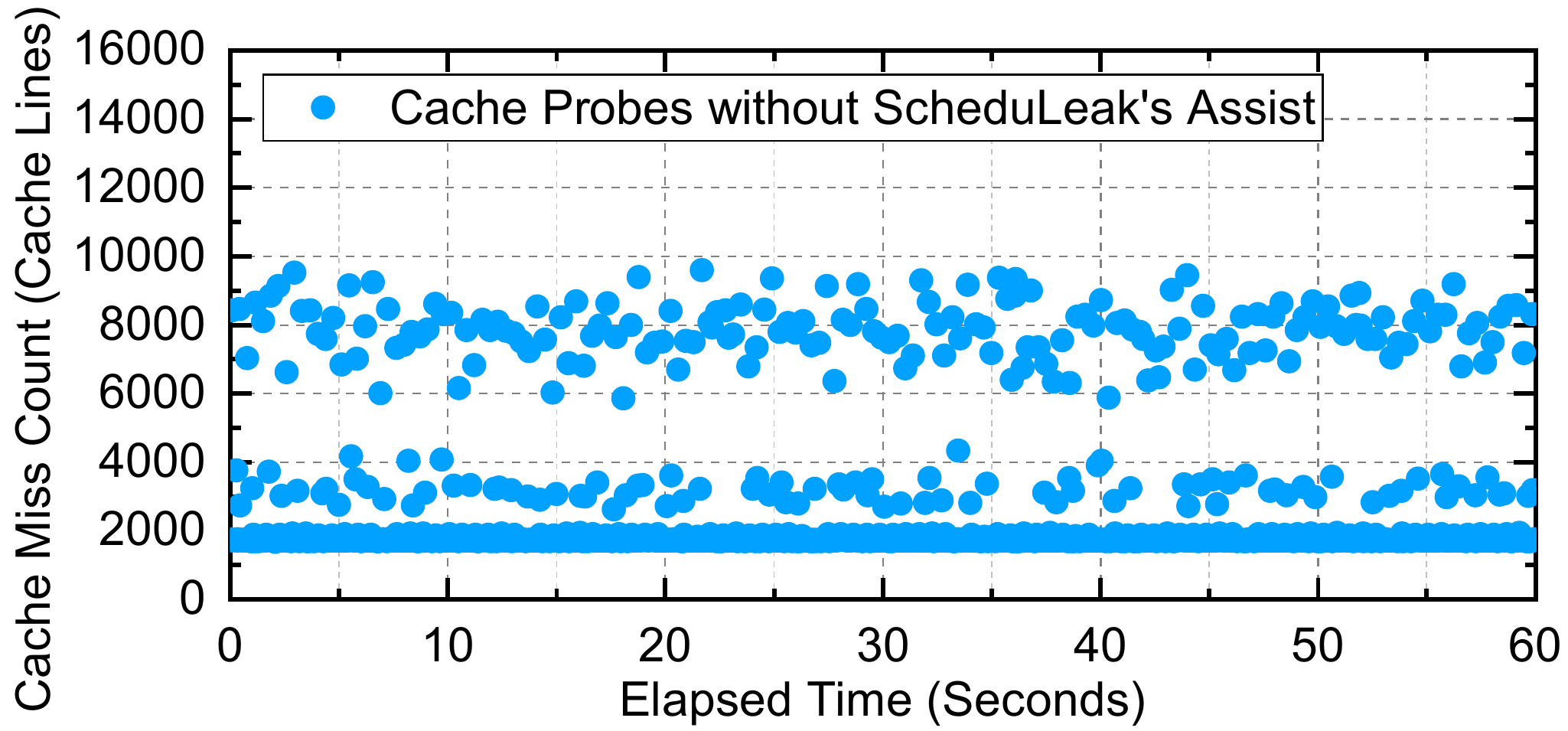}
        \vspace{-0.2\baselineskip}
        \caption{Attack without \ATTCKNF's assist.}
    \end{subfigure}%
    \begin{subfigure}{0.36\textwidth}
        \centering
        \includegraphics[width=0.98\columnwidth]{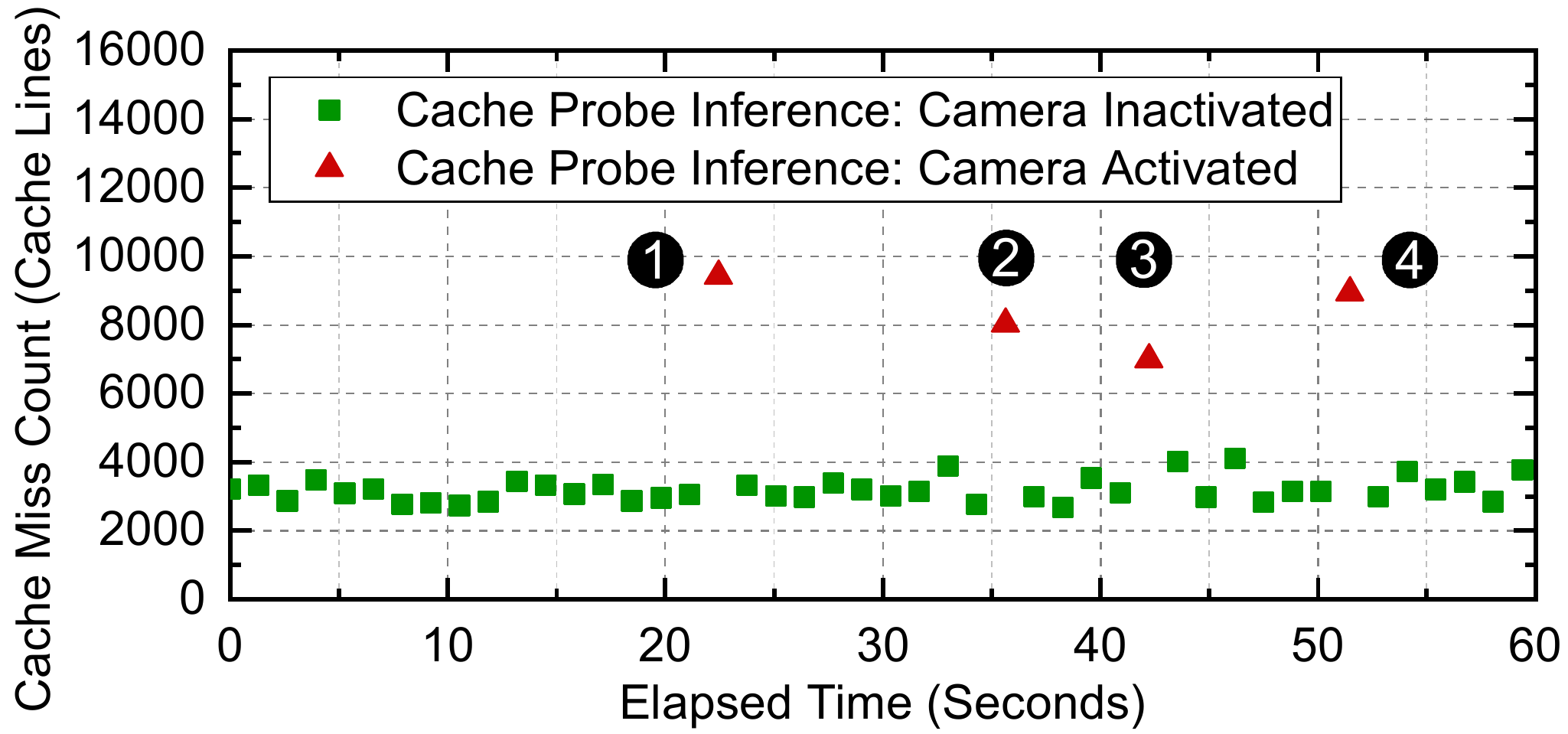}
        \vspace{-0.2\baselineskip}
        \caption{Attack with \ATTCKNF's assist.}
    \end{subfigure}
	\begin{subfigure}{0.27\textwidth}
    	\centering
  		\includegraphics[width=0.98\columnwidth]{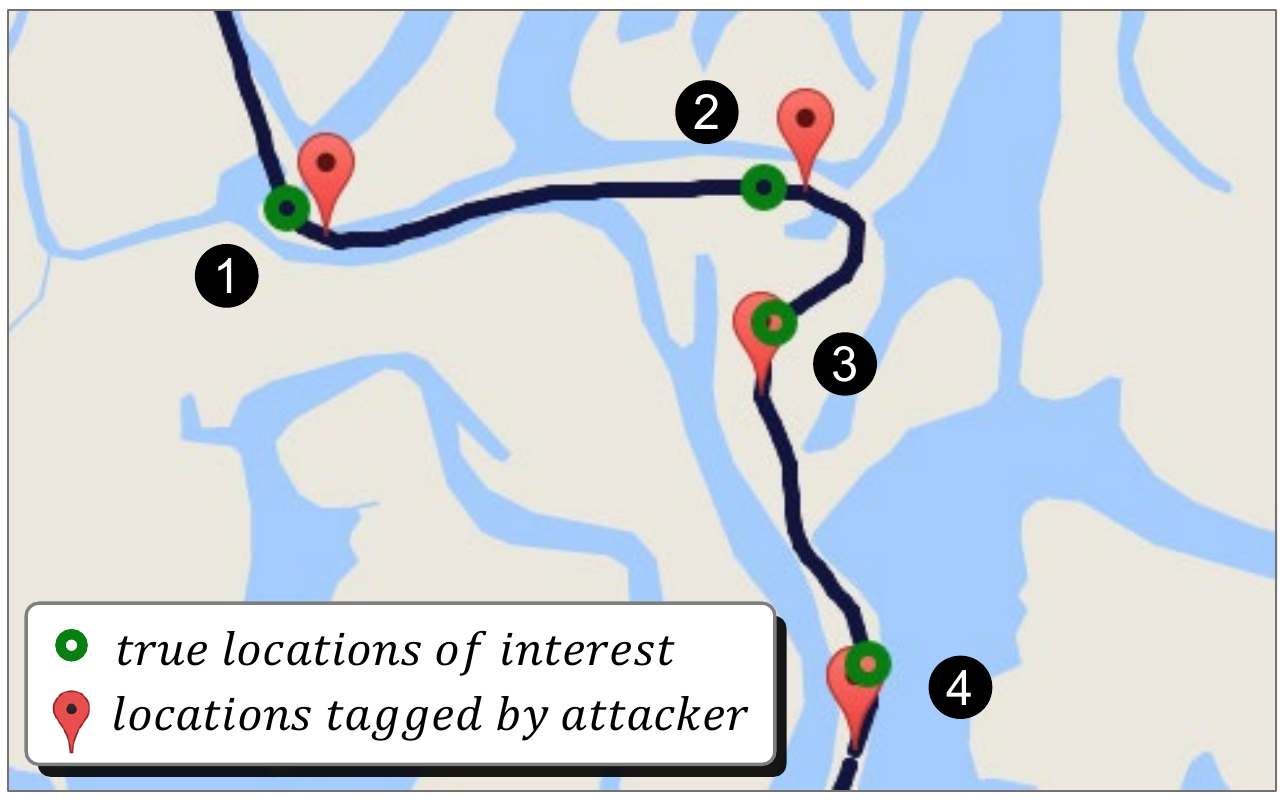}
		\vspace{-1.2\baselineskip}
		\caption{HIL simulator recorded data.}
		\label{fig:cache_attack_gps_trace}
	\end{subfigure}
    \vspace{-0.4\baselineskip}
    \caption{Results of the cache-timing side-channel attacks 		
    in Section~\ref{sec:cache_timing_attack}. 
    (a) demonstrates that a random mechanism launching the attack at arbitrary instants will lead to many indistinguishable cache usage results. (b) shows a successful attack in which four camera activation events (numbered by $1$ to $4$) are identified from the cache probes using precise time information (inferred by \ATTCKNF). (c) visualizes the UAV's trajectory (bold line), true locations-of-interest (green circles) and the attacker's inference (red pins) for the attack (b). The result shows that the attacker's inference matches the ground truth.
}
    \label{fig:cache_attack_probes}
    \vspace{-1\baselineskip}
\end{figure*}

Before evaluating performance of the introduced algorithms, we first aim to evaluate the feasibility of such algorithms on realistic platforms in this section.
The \ATTCKNF algorithms are implemented on two operating systems with a real-time scheduling capability: \ci Real-Time Linux~\cite{RTLinux} and \cii \text{FreeRTOS}~\cite{FreeRTOS}.
%
In what follows, two attack cases are presented. 
They benefit from the information obtained by the proposed algorithms 
and utilize such information to accomplish their primary attack goals.
The demo videos for these attack cases can be found at \url{https://scheduleak.github.io/}.


\subsection{Overriding Control Signals}
\label{sec:take_over_the_control}
\noindent {\bf Attack Scenario and Objective:}
A large number of real-time control systems encapsulate subsystems that control actuators.
For instance, in modern automotive systems, the engine control unit (ECU) controls the valve in the electronic throttle body (ETB) to enable electronic throttle control (ETC).
In most unmanned drones, the flight controller manages the rotary speed of the motors via the electronic speed controller (ESC).
In these systems, the actuation signals such as PWM signals are periodically updated to guarantee a fast and consistent response for the control mission. 

Let's consider an attacker who wants to be able to stealthily override the control in such systems -- for the purpose of 
bad control by causing 
misbehavior or even taking over the control of the system for a short time span.
To do so, the attacker gets into the system as a malicious task and tries to override the control signals.
A brute force strategy of excessively overriding the control signals will not work in this scenario because its high attack overhead can cause other real-time tasks to miss their deadlines and lead to a system crash.
In this case, knowledge of exact timing when the control signals are updated and overriding them at the right instants allow the attacker to effectively take control with a low overhead.

\noindent {\bf Implementation:}
We implement this attack on a custom rover. 
Its control system is built with a Raspberry Pi 3 Model B board. 
A Navio2 module board that encapsulates various inertial sensors is attached to the Raspberry Pi board.
The system runs Real-Time Linux (\ie Raspbian, kernel 4.9.45 with PREEMPT\_RT patch) with Ardupilot~\cite{Ardupilot} autopilot software suite (one of the most popular open-source code stack in the remote and autonomous control communities).
It consists of a set of real-time and non-real-time tasks to perform control-related jobs such as refreshing GPS coordinates, decoding remote control commands, performing PID calculation and updating output signals.
One of the tasks periodically updates the PWM values, with a period of $20ms$, for steering and throttle. The updates are sent over Serial Peripheral Interface (SPI) to the Navio2 module that outputs the PWM signals to a servo and a ESC.
Figure~\ref{fig:exp_pwm_override}(a) shows an illustration of the PWM output channels working under normal circumstances.

In this attack, we assume that the attacker has access to a low-priority, periodic task (as the observer task, $p_o=50ms$) and a non-real-time Linux process (for launching the PWM overriding attack).
The attacker's ultimate objective is to override the control signals updated by the victim task (\textit{i.e.,} the $50Hz$ periodic task).
In this implementation, the observer task uses a system call, $\mathtt{clock\_gettime()}$, to obtain clock counts (in nanoseconds) from $\mathtt{CLOCK\_MONOTONIC}$.
Time measurement is further rounded up to microseconds when running the ScheduLeak algorithms since all task parameters are multiples of $1us$ in Ardupilot.
Once the victim task's initial offset is determined, the attacker engages the non-real-time process to issue the PWM updates over the same interface that the victim task uses. Note that this is possible due to a lack of authentication between the Raspberry Pi board and the Navio2 module by design.
This process keeps track of time by using $\mathtt{clock\_gettime()}$ and issues two PWM updates (one for the steering and one for the throttle) whenever it determines that it has passed a victim task's arrival instant (\textit{i.e.,} $t-\hat{a_v}\ mod \ p_v \geq 0$, where $t$ is the present time and $\hat{a_v}$ is the inferred victim task's initial offset). The process remains idle between two PWM updates to reduce the attack footprint.

\noindent {\bf Attack Results:}
Figures~\ref{fig:exp_pwm_override}(b) and \ref{fig:exp_pwm_override}(c) show that the PWM output may be overridden using a different value to the PWM hardware.
However, without exact schedule information, the attacker can only periodically send the updates with a randomly selected initial offset (Figure~{\ref{fig:exp_pwm_override}}(b)).
The random initial offset can be any point in the $20ms$ period.
From our experiments, only the attack with an initial offset in the range between $a_v$ and $a_v + 8.3ms$ 
can produce an effective override of the steering and throttle controls. As a result, the attacker has a chance of $41.5\%$ to select a valid initial offset and lead to an effective attack.

On the other hand, the attacker, after launching the \ATTCKNF attack and knowing exactly when the victim task arrives,
can carefully issue PWM update {\em right after the original update} to override the PWM output (Figure~\ref{fig:exp_pwm_override}(c)).
In this case, the attacker firstly runs the ScheduLeak algorithms in the observer task, yielding $0.9985$ for the inference precision ratio (for inferring the victim task's initial offset) in a duration of 1 second.
This allows the attacker to launch the PWM overriding attack in the non-real-time process with the precise inference of the victim task's initial offset. Note that an attacker's PWM update attempted at a victim task's arrival instant is executed after the victim task's job is finished (and hence after the original PWM update) since the non-real-time process has a priority lower than the victim task.
Consequently, the attacker can take over control of the steering and throttle.
By probing the PWM signals, we observe that the overridden PWM signals are active $85\%$ of the time.
As a result, we see that the rover no longer responds to the original control. 
Instead, the rover is driven by the attacker's commands.
Since the attacker's task remains idle between two PWM updates, it takes up CPU utilization as small as $2.6\%$.

\subsection{Inferring System Behaviors}
\label{sec:cache_timing_attack}


\noindent {\bf Attack Scenario and Objective:}
Let's consider a UAV system 
executing a surveillance mission. 
It captures high resolution images when flying over locations of high-interest.
In this case, the attacker's goal is to extract the locations targeted by the UAV.
The strategy is to monitor when the surveillance camera on the UAV is switched to a execution mode in which high-resolution images are being processed.
This can be done by exploiting a cache-timing side-channel attack to gauge the coarse-grained memory usage behavior of the task that handles the images.
A high cache usage by this task would indicate that a high-resolution image is being processed; otherwise it would use less cache memory.
However, a random sampling of the cache will result in noisy (and often useless) data since there exist other tasks in the system that also use the cache. 
In contrast, knowing when the task is scheduled to run allows the attacker to execute prime and probe attacks \cite{osvik2006cache, page2002theoretical} very close to the targeted task's execution.
%

\noindent {\bf Implementation:}
This attack is implemented in a hardware-in-the-loop (HIL) simulation with a Zedboard running \text{FreeRTOS} that simulates the control system on a UAV.
The system consists of an image processing task (the victim task, $p_v=33ms$) handling photos at a rate of $30Hz$ 
and four other tasks (unknown to the attacker) -- all running in a periodic fashion.
The victim task 
processes a large size of data when the UAV reaches a location of interest on a preloaded list.
Other tasks consume differing amounts of memory.
In this case, we assume that the attacker enters the system as the lowest-priority periodic task, $p_o=40ms$. 
The attacker uses this task for both running the ScheduLeak algorithms and carrying out the cache-timing side-channel attack.
The attacker's final goal is to observe the victim task's memory usage and learn the system behavior.

\noindent {\bf Attack Results:}
First, we consider an attacker who does not employ a \ATTCKNF attack. 
The attacker launches the cache-timing side-channel attack during every period to try and estimate the cache usage of the victim.
As shown in Figure~\ref{fig:cache_attack_probes}(a), this produces many cache probes and it is hard to distinguish the cache usage of the victim task from other tasks.
This results in an unsuccessful attack since no usage patterns from the victim task can be identified. 

Next, let's consider the case in which the attacker leverages the \ATTCKNF attack.
In this case, the algorithms yield an inference precision ratio of {$0.99$} 
within a window of $3 \cdot LCM(p_o,p_v)$ (\ie 4 seconds).
Then, the attacker is able to launch the cache-timing side-channel attack right before and after the victim is executed and skip those instants that are irrelevant.
Figure~\ref{fig:cache_attack_probes}(b) shows the result of the precise cache probe against the victim task.
We see that the attack greatly reduces the noise caused by other tasks ($96.9\%$ of the cache probes are omitted) and is able to precisely identify the victim task's memory usage behavior.
%
As a result, four camera activation instants can be identified from the spikes (red triangular points) shown in Figure~\ref{fig:cache_attack_probes}(b).
When coupled with the flight route information that the attacker obtains through other measures,
it becomes possible to infer the locations of high-interest, as shown in 
Figure~\ref{fig:cache_attack_probes}(c). 

\section{Performance Evaluation \\and Design Space Exploration}
\label{sec::eval}

\begin{figure*}
\centering
\begin{minipage}{.32\textwidth}
  \centering
  \includegraphics[width=0.975\columnwidth]{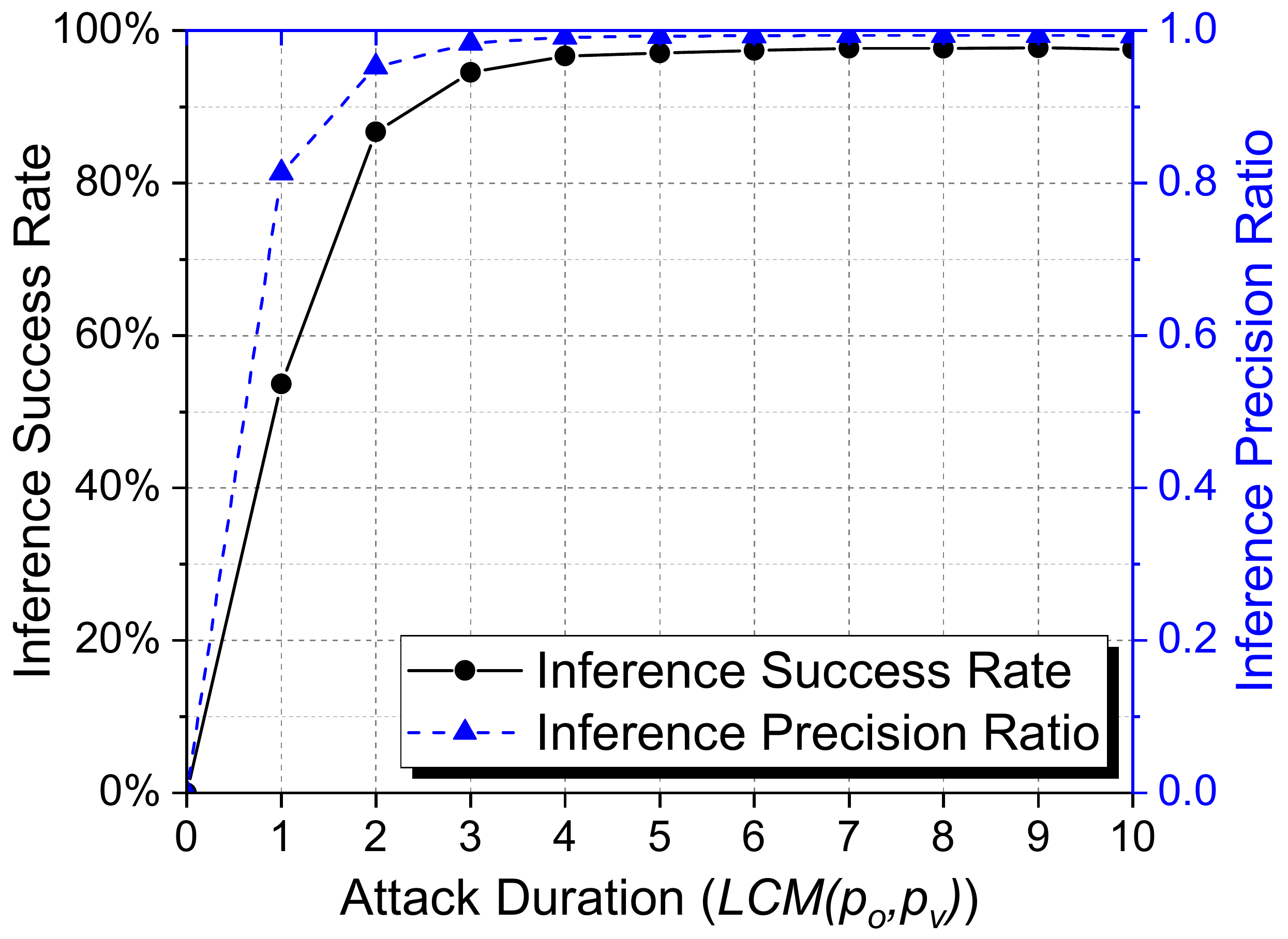}
\vspace{-0.6\baselineskip}
\caption{The results of varying attack duration. 
It indicates that longer attack durations can increase the chance of success and yield better inference precision. The points are connected only as a guide. }
\label{fig:timeVsSuccessRate}
\end{minipage}%
\hfill
\begin{minipage}{.32\textwidth}
  \centering
  \includegraphics[width=0.98\columnwidth]{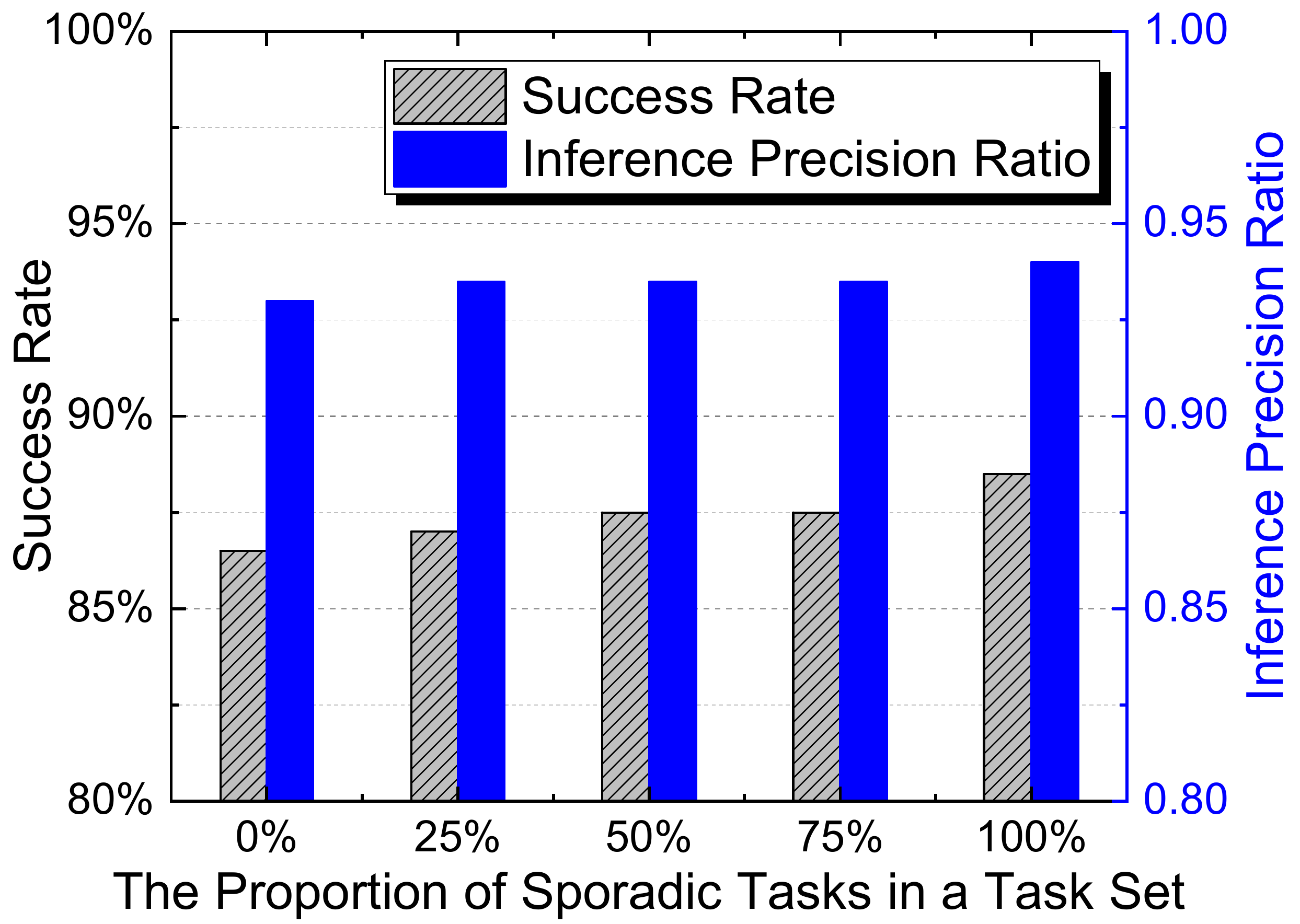}
\vspace{-0.5\baselineskip}
	\caption{The impact of sporadic tasks. 
It indicates that the algorithms perform better with sporadic tasks, with a (slightly) ascending trend as the proportion of sporadic tasks increases.}
\label{fig:sporadicVsPeriodic}
\end{minipage}
\hfill
\begin{minipage}{.32\textwidth}
  \centering
  \includegraphics[width=0.98\columnwidth]{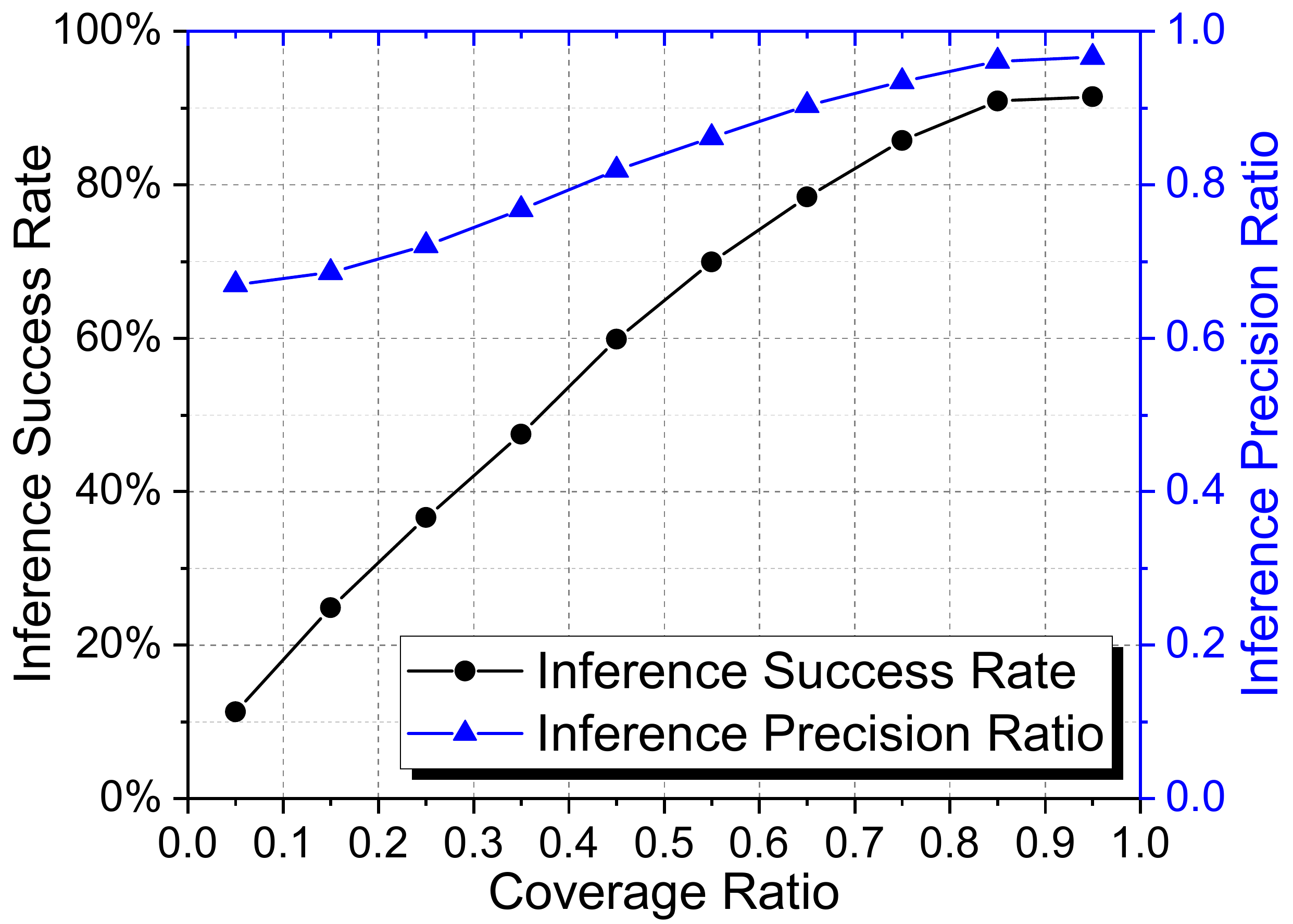}
\vspace{-0.5\baselineskip}
\caption{The performance of the algorithms when $\mathbb{C}(\tau_o,\tau_v) < 1$. 
Round and triangular points represent the inference success rate and the inference precision ratio, respectively.
}
\label{fig:coverageRatio}
\end{minipage}
\end{figure*}

\begin{figure*}
\centering
\begin{minipage}{.3\textwidth}
  \centering
  \includegraphics[width=.9\linewidth]{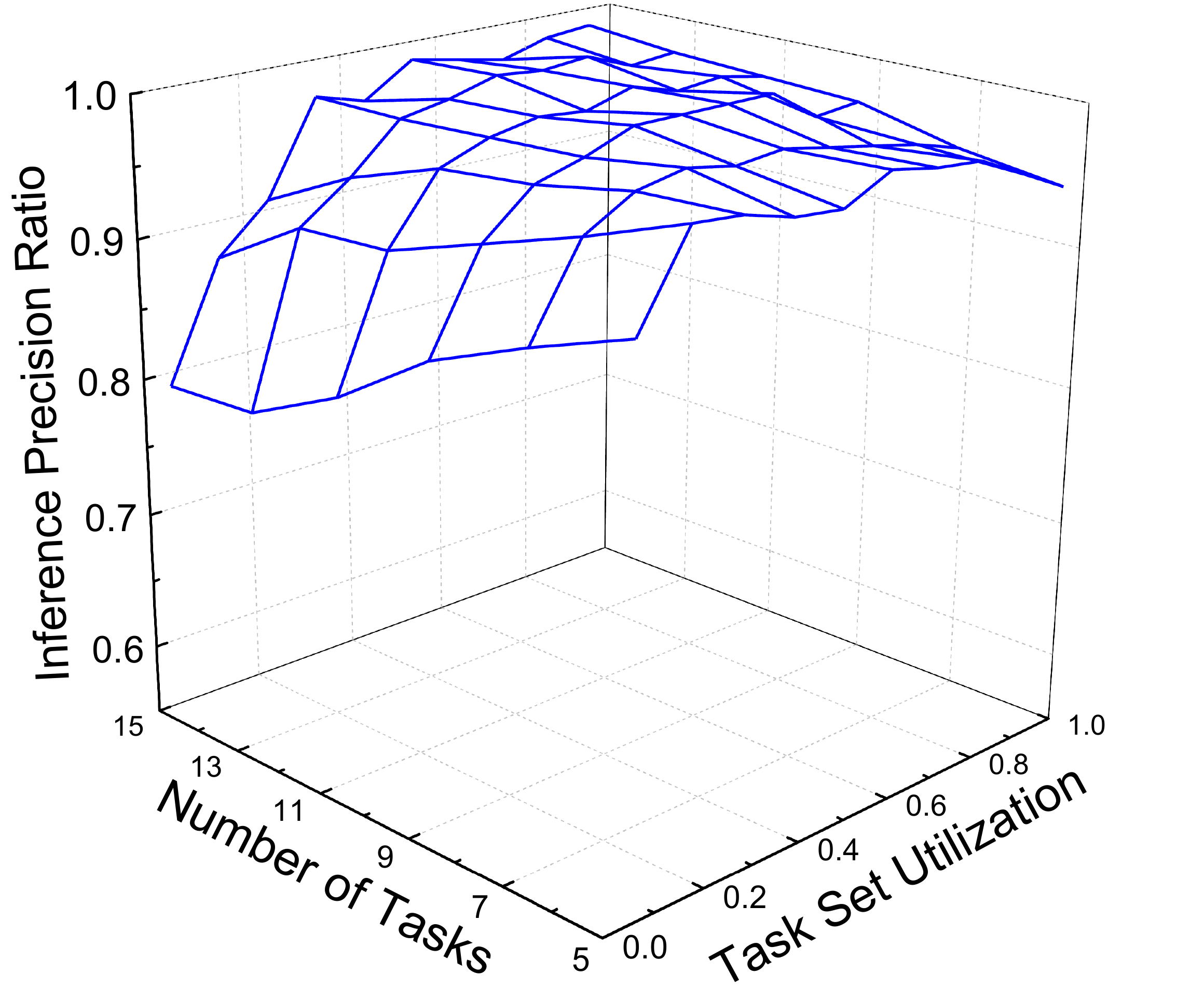}
  \caption{The impact of the number of tasks and the task set utilization. 
It shows that the algorithms perform better with small number of tasks and high task set utilization.}
  \label{fig:numOfTasksVsUtil}
\end{minipage}%
\hfill
\begin{minipage}{.68\textwidth}
  \centering
  \begin{subfigure}[t]{0.49\textwidth}
  \centering
  \includegraphics[width=.9\linewidth]{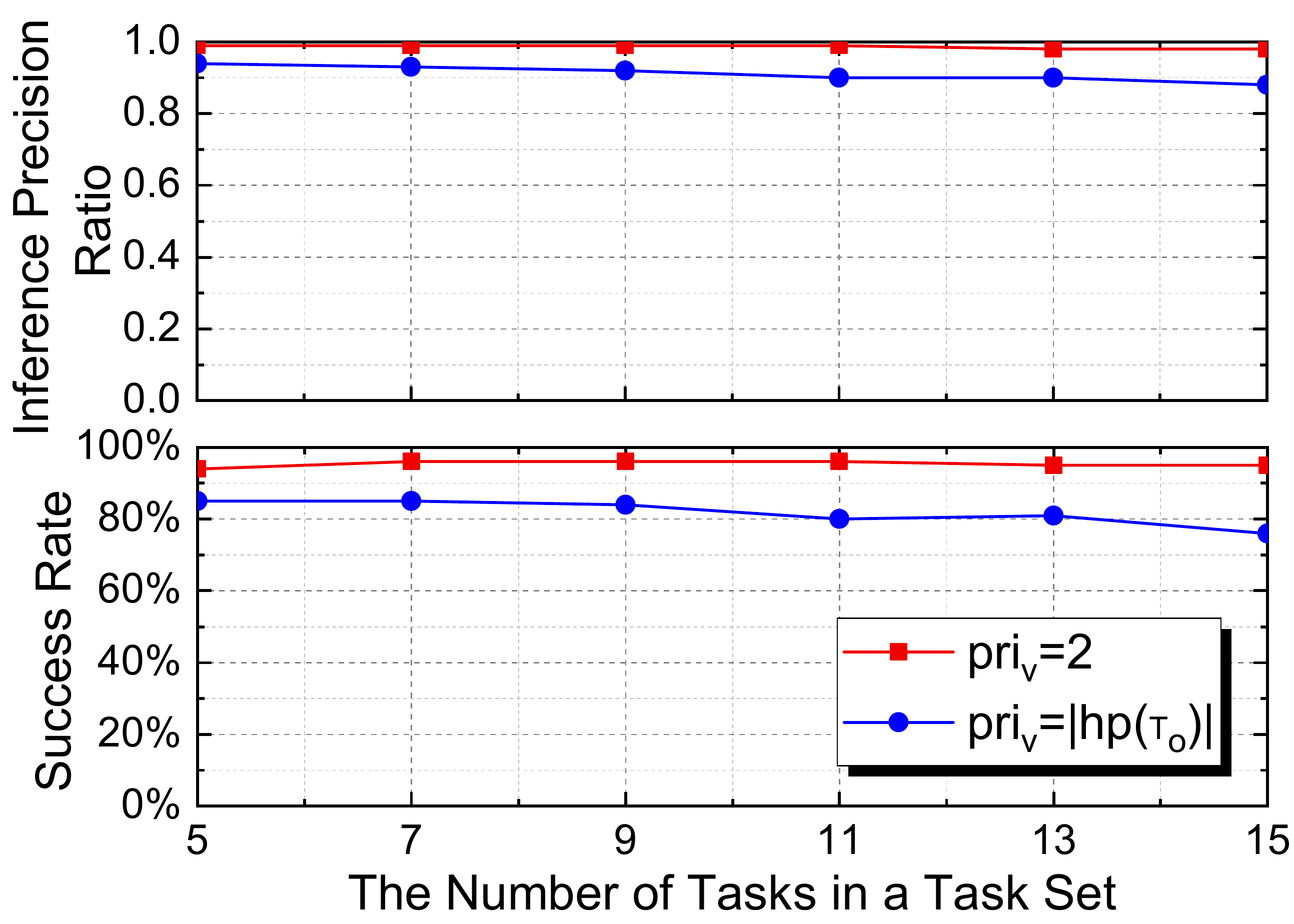}
  \vspace{-0.3\baselineskip}
  \caption{Grouped by the number of tasks.}
  \label{fig:victimTaskPriority_numOfTasks}
  \end{subfigure}
  \hfill
  \begin{subfigure}[t]{0.49\textwidth}
  \centering
  \includegraphics[width=.9\linewidth]{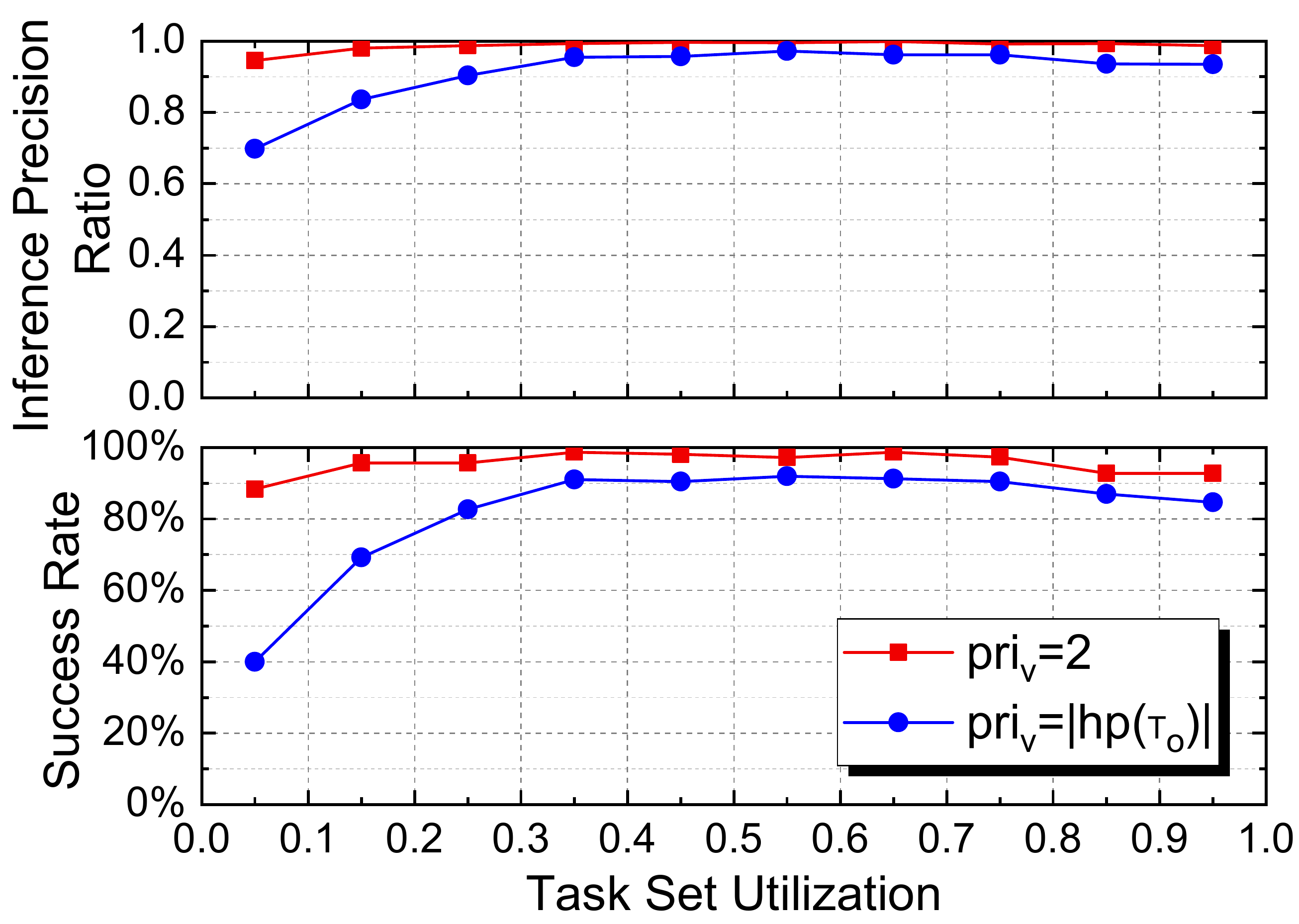}
  \vspace{-0.3\baselineskip}
  \caption{Grouped by task set utilization.}
  \label{fig:victimTaskPriority_util}
  \end{subfigure}
  \vspace{-0.5\baselineskip}
  \caption{The impact of the victim task's position in a task set. It suggests that a victim task with higher priority makes it hard for the algorithms to make a correct inference. This result stands throughout different number of tasks in a task set as well as different task set utilization. Also, a high priority victim task with low task set utilization reduces the inference performance. This explains the huge drop in Figure~{\ref{fig:numOfTasksVsUtil}}.}
  \label{fig:victimTaskPriority}
\end{minipage}
\vspace{-1\baselineskip}
\end{figure*}
\subsection{Evaluation Setup}
\label{sec::eval_setup}

%
We test our algorithms with randomly generated synthetic
task sets. 
The task sets are grouped by CPU utilization from $[0.001+0.1\cdot x,  0.1+0.1\cdot x]$ where $0
\leq x \leq 9$. 
Each utilization group
consists of $6$ subgroups that have a fixed number of tasks ($5, 7, 9, 11,
13, 15$). Each subgroup contains $100$ task sets. 
In each task set, $50\%$ of the tasks are generated as periodic tasks ($3,
4, 5, 6, 7, 8$ periodic tasks for each subgroup respectively) while the
rest of the tasks are generated as sporadic tasks. 
The task periods are randomly drawn from $[100,1000]$ and we assume that the attacker has access to the system time with a resolution of $1$.
The task initial offset 
is randomly selected from $[0, p_i)$.
In the case of sporadic tasks, we take the generated task period as the minimum inter-arrival
time.
The task priorities are assigned using the rate-monotonic algorithm \cite{LiuLayland1973}.
We only pick those task sets that are schedulable.

The observer task and the victim task are assigned when generating the task sets. 
In simulations, we consider a periodic observer task because it represents the worst case attack scenario for the adversary, as discussed in Section~{\ref{sec::determine_capability}}.
Since only the tasks with higher priorities influence the observations, 
we skip the generation of lower-priority tasks $lp(\tau_o)$.
Thus, the observer task always has the lowest
priority (\ie $pri_o=1$) in these generated task sets. For the victim task, two
conditions are considered: \ci $pri_v=2$ and \cii $pri_v=\left | hp(\tau_o)
\right |$. This is to test the two boundary conditions. 
Further,
we set the coverage ratio to be $\mathbb{C}(\tau_o, \tau_v) \geq 1$ when
generating the task sets (except for evaluating the impact of the coverage
ratio), to evaluate whether the algorithms can truly
produce confident inferences while the attacker has theoretical guarantees
of the attack capability (\ie having full coverage of all $p_v$ time columns,
as per Theorem~\ref{th:full_coverage}). 
The maximum construction duration $\BUDGET$ is set as per Section~{\ref{sec::sigma}}. 
Thus, $\BUDGET=GCD(p_o,p_v)$.

%
For varying the execution times of the tasks and adding jitter to the inter-arrival times (for the sporadic tasks), we use the normal and Poisson distributions respectively.
Note that Poisson distribution is used for inter-arrival
time variation because the probability of each occurrence (\ie each arrival of
the sporadic task) is independent in such a distribution model.
First, a schedulable task set is generated (using the aforementioned parameters). 
Then, for a task $\tau_i$, 
the average execution time is computed by $wcet_i\cdot 80\%$. 
Next, we fit a normal distribution $\mathcal{N}(\mu, \sigma^2)$ for the task
$\tau_i$. We let the mean value $\mu$ be $wcet_i\cdot 80\%$ and find the
standard deviation $\sigma$ with which the cumulative probability $P(X \leq
wcet_i)$ is $99.99\%$. As a result, such a normal distribution produces
variation such that $95\%$ of the execution times are within $\pm 10\% \cdot wcet_i$. 
To ensure that the task set remains schedulable, we adjust the maximum modified 
execution time to be equal to WCET if it
exceeds WCET. 
For sporadic tasks, the average inter-arrival time is computed by $p_i \cdot
120\%$. 
We use a Poisson distribution with $p_i \cdot 120\%$ as its mean value to generate the varied inter-arrival times during the
simulation. Similarly, so as to not violate the given minimum inter-arrival time for
a sporadic task, we regenerate the modified inter-arrival time if it drops below
$p_i$. 
%


\subsection{Results}
\label{subsec::results}

\subsubsection{Attack Duration}
\label{sec::exp_attack_duration}
Our first goal is to understand the effects of how long attacks last. 
Recall that the coverage of the schedule ladder diagram repeats every
$LCM(p_o,p_v)$ (Observation~\ref{observation:lcm}). 
Therefore, we use $LCM(p_o,p_v)$ as the unit of time 
to evaluate the algorithms. 
Taking the Ardupilot software 
as an example, the largest $LCM$ of any real-time task (\ie a AP\_HAL thread) pairs is $20ms$. While $LCM(p_o,p_v)$ varies system to system, this gives us an insight into the scale of $LCM(p_o,p_v)$.
In this experiment, we generate task sets as explained in 
Section~\ref{sec::eval_setup} and run the \ATTCKNF algorithms with a fixed
duration of $10 \cdot LCM(p_o,p_v)$ for every task set. Figure~\ref{fig:timeVsSuccessRate} 
shows the results of this experiment.
In Figure~\ref{fig:timeVsSuccessRate}, each point of the inference precision
ratio is the mean of the individual inference precision ratios of $12000$ task
sets for a given attack duration. The results suggest that the longer the
attack is sustained, the higher success rate and precision ratio the algorithms can
achieve. This is because a longer attack time means more execution intervals are
reconstructed by the observer task.
%
On the other hand, both success
rate and precision ratio plateau after $5 \cdot LCM(p_o,p_v)$ with the
success rate and the precision ratio higher than $97\%$ and $0.99$
respectively. This shows that the proposed algorithms can produce inference
with precision in a very short time and the additional gains obtained from running longer are
minuscule. For this reason, \emph{we evaluate
the algorithms with a duration of $10 \cdot LCM(p_o,p_v)$ for the rest of the
experiments below}.

\subsubsection{The Number of Tasks and Task Set Utilization}
\label{sec::exp_the_impact_of_num_tasks_and_util}
Figure~\ref{fig:numOfTasksVsUtil} displays a 3D 
graph that shows the averaged inference precision ratio for each combination of
the number of tasks and the task utilization subgroup. The results suggest
that \ci the inference precision ratio decreases as the number of tasks in a
task set increases and \cii the inference precision ratio increases as the task
set utilization increases. The worst inference precision ratio happens when
there are $15$ tasks in a task set with the utilization group $[0.001, 0.1]$
-- these are boundary conditions for both the number tasks and the utilization in
this experiment. The impact of the number of tasks is straightforward as
having more tasks in $hp(\tau_o)$ means that $\tau_o$ will be
preempted more frequently. This makes it hard for the observer task to eliminate the
false time columns. 
For the impact of the task set utilization,
a low utilization value implies that the execution times of the tasks are small and
there exists a lot of gaps in the schedule. Hence, the observer may get many small
and scattered intervals. Since we let the algorithms pick the largest interval to infer
the true arrival column, multiple small intervals
are problematic -- the algorithm has a hard time picking the right interval that contains
the true arrival. Hence errors are compounded.

\subsubsection{Priority of the Victim Task}
We analyze the impact of the victim task's priority in a task
set.
From Section~\ref{sec::eval_setup}, we consider two
boundary conditions for the victim task's position: \ci $pri_v=2$ and \cii
$pri_v=\left | hp(\tau_o) \right |$.
Figures~\ref{fig:victimTaskPriority}(a) and
\ref{fig:victimTaskPriority}(b) present the experiment results for the
two conditions.  Figure~\ref{fig:victimTaskPriority}(a) shows that the
huge drop in Figure~\ref{fig:numOfTasksVsUtil} 
(as the number of tasks increases)
is mainly caused by the condition $pri_v=\left | hp(\tau_o) \right |$.
Figure~\ref{fig:victimTaskPriority}(b) also shows the similar indication that the drop in low utilization groups in
Figure~\ref{fig:numOfTasksVsUtil} is a result of the condition $pri_v=\left |
hp(\tau_o) \right |$.  It's worth noting that, since we use the rate-monotonic
algorithm to assign the priority, $pri_v=2$ means that $\tau_v$ has a
large period, hence potentially has greater execution time. 
It benefits the algorithms as we pick the largest interval to make an inference in the final step. 
%

\subsubsection{Sporadic and Periodic Tasks}
We examine the impact of the mix of sporadic and periodic tasks.
We generate task sets with $0\%$, $25\%$, $50\%$, $75\%$ and $100\%$ sporadic tasks in a task set.
The rest of the tasks in a task set are periodic tasks.
Comparing the result of all periodic tasks 
and the result of all sporadic tasks shown in Figure~\ref{fig:sporadicVsPeriodic}, 
we find that the {\em algorithms perform better with more sporadic tasks}.
It shows an ascending trend as the proportion of sporadic tasks increases.
However, the change in the performance is less than $1\%$, which is subtle.
Hence, our inference algorithms are fairly agnostic to the actual mix of sporadic/periodic tasks
in the system.

\subsubsection{Coverage Ratio and The Maximum Reconstruction Duration}
\label{sec:::exp_coverage_ratio}
The experiments above show that the algorithms can reach certain inference
success rates and precision when $\mathbb{C}(\tau_o,\tau_v) \geq 1$ and
$\BUDGET=GCD(p_o,p_v)$.  However, attackers may face a victim system where
$\mathbb{C}(\tau_o,\tau_v) < 1$.  That is, the observer task's execution is not
guaranteed to appear in all $p_v$ time columns. To evaluate the performance of the
algorithms against such a case, we generate task sets with $0 <
\mathbb{C}(\tau_o,\tau_v) < 1$ (thus $\BUDGET=e_o$) and run the algorithms for a
duration of $10 \cdot LCM(p_o,p_v)$.  In this experiment, task sets are grouped
by coverage ratio from $[0.001+0.1\cdot x,  0.1+0.1\cdot x]$ where $0 \leq x
\leq 9$.  Figure~\ref{fig:coverageRatio} shows the results.  It suggests that
the attacker may fail to completely infer the victim task's initial offset when the
coverage ratio is low. 
Yet, the algorithms can still succeed in some cases due to the fact that Theorem~{\ref{th:blank_columns}} holds even with a low coverage ratio.
When the observer has about half coverage of the
time columns (the group of  $[0.401,0.5]$), it yields $59.9\%$ in success rate
and $0.819$ for the averaged inference precision ratio. 
As more time columns
are observed by the observer task, the precision and success rate increase.
This is because higher coverage ratios give the algorithms a higher chance
to capture the true arrival column and remove others. As a result, the inference
success rate is about proportional to the coverage ratio.

\section{Discussion -- Potential Defense Strategies}
\label{sec::discussion}

To defend against the proposed attack algorithms, one strategy could be to enforce a low coverage ratio between any low priority task and the critical real-time task by adjusting the task parameters.
This reduces the attacker's observability/capability (based on results from Section~{\ref{sec:::exp_coverage_ratio}}).
Furthermore, carefully designing and employing a harmonic taskset may also reduce ScheduLeak's inference precision since it creates multiple candidates in the last step of the algorithms.
However, any change in the task parameters must fulfill both real-time requirements as well as the required performance.
Thus, changing the task parameters may not always be applicable in real-time systems especially the legacy systems that are already deployed.

Since the proposed algorithms rely on the repeating patterns of the victim task, a potential countermeasure is to perturb the periodicity of the system schedule.
Yet, the measure will not be trivial due to the real-time constraints of real-time tasks.
A careless solution can easily cause some real-time tasks to miss their deadlines and lead to a system failure.
A randomization protocol for a rate-monotonic scheduler presented by Yoon {\etal}~{\cite{2016:taskshuffler}} is a good attempt on removing the scheduler side-channel for RTS.
However, their work is not applicable in our case because they only focus on the systems with all periodic tasks while our work is feasible on the systems with both periodic and sporadic tasks (which is the case in most real-time control systems).
Therefore, an effective solution would need to consider covering both task types.

\section{Related Work}
\label{sec:related}


The problem of information leakage via side-channels has been well studied in the literature. 
For instance, it has been shown that cache-based side-channels can be invaluable for information leakage \cite{Hu92, 
kelsey1998side, osvik2006cache, page2002theoretical}. With the advent of multi-tenant public clouds, cache-based side-channels and their defenses have received renewed interest (\eg \cite{Zhang2012, Apecechea_finegrain, ristenpart2009hey, Wang:2007}).
%
Other types of side-channels such as differential power analysis \cite{Jiang2014}, electromagnetic and frequency analysis \cite{Tiu05anew,Agrawal:2002:ES} have also been studied. Our focus here is on scheduler side-channels in real-time systems.

There has also been some work on information flow via schedulers.
The problem where two tasks leak private information by using a covert channel is studied \cite{ghassami2015capacity, embeddedsecurity:son2006}.
V\"{o}lp \emph{et al.} \cite{embeddedsecurity:volp2008,embeddedsecurity:volp2013} examined covert channels between different priorities of real-time tasks and proposed solutions to avoid such covert channels.
The methodologies for quantifying information leakage in schedulers are also studied \cite{kadloor2013,GongK14}.
While the previous works focused on covert channels in some schedulers, our focus is on novel side-channels in real-time schedulers where an unprivileged low-priority task can infer the execution timing behaviors of high-priority real-time task(s).
Also, in contrast to covert channels that rely on actively preempting real-time tasks, the side-channel in our work does not violate any real-time constraints and the observer task only observes its own behavior.

The integration of security into real-time schedulers is a developing area of research. Mohan \emph{et al.} \cite{embeddedsecurity:mohan2014} offered a consideration of real-time system security requirements as a set of scheduling constraints and introduced a modified fixed-priority scheduling algorithm that integrates security levels into scheduling decisions.
Pellizzoni \emph{et al.} \cite{embeddedsecurity:mohan2015} extended the above scheme to a more general task model and also proposed an optimal priority assignment method that determines the task preemptibility. 
Some researchers also have focused on defense techniques for real-time systems (\eg \cite{embeddedsecurity:mohans3a2013,yoonsecurecore2013,Zadeh:2014,YoonMHM:2015,xie2007schedulesecurity,lin2009rtssecurity,trilla2018cache}).
However, these solutions do not protect the systems from the \ATTCKNF attack. 

The most closely related solution is to adopt a randomization technique to obfuscate the schedule.
Yoon \etal~{\cite{2016:taskshuffler}} introduced a randomization protocol for a preemptive, fixed-priority scheduler that works with only (fully) periodic tasks. Kr{\"u}ger \etal~ {\cite{kruger2018vulnerability}} built upon this by proposing
an online job randomization algorithm for time-triggered systems.
Nevertheless, these solutions are not applicable to most real-time systems in which a preemptive, fixed-priority scheduler supports both periodic and sporadic real-time tasks. This leaves those systems still vulnerable to our \ATTCKNF attack.


\section{Conclusion}
\label{sec::concl}
Successful security breaches in control systems (including cyber-physical systems) with real-time properties can have catastrophic effects. 
In many such systems, knowledge of the precise timing information
of critical tasks could be beneficial to adversaries.
Our work in this paper demonstrates how to capture 
this schedule timing information in a {\em stealthy} manner 
-- \ie without being detected or causing any perturbations to the 
original system. Designers of such systems now need to be cognizant 
of such attack vectors and design the system to include countermeasures
that can thwart potential intruders. The end
result is that real-time systems
can be more robust to security threats overall.

\setcounter{section}{0}
\section*{Appendix}

\subsection{Algorithm for Reconstructing An Execution Interval}
\label{sec::reconstruct_an_execution_interval_algo}
\begin{algorithm}
\caption{Reconstructing An Execution Interval $\mathbb{E}(e_o, e_o', \BUDGET)$}
\label{algo:construct_ei}
\begin{algorithmic}[1]
\Statex \{$GT: global \ timer \ (system \ timer)$\}
\Statex \{$e_o: the \ worst \ case \ execution \ time \ of \ \tau_o $\}
\Statex \{$e_o': remaining \ execution \ time \ of \ present \ job \ of \ \tau_o $\}
\Statex \{$\BUDGET: maximum \ reconstruction \ duration \ in \ a \ period$\}
\Statex \{$t_{stop}: stop \ time \ when \ \BUDGET \ is \ met $\}
\Statex \{$t_{0}, t_{-1}: present \ and \ last \ time \ stamps$\}
\Statex \{$t_{begin}, t_{end}: start, \ end \ time \ of \ the \ detected \ interval$\}
%
\State $t_0 = GT$ 
\State $t_{begin} = t_0$
\State $t_{stop} = t_{begin} + e_o' - (e_o - \BUDGET)$
\State $duration = 0$
\While {$duration \leq loop \ execution \ time \ unit \ \textbf{and} \ t_0 < t_{stop}$}
\State $t_{-1} = t_0$
\State $t_0 = GT$
\State $duration = t_0 - t_{-1}$
\EndWhile 
\If {$duration > loop \ execution \ time \ unit$}
\State $t_{end} = t_{-1}$
\Else
\State $t_{end} = t_0$
\EndIf
\State $e_o' = e_o' - (t_{end} - t_{begin})$ \\
\Return \{$t_{begin}, \ t_{end}, \ e_o' $\}
\end{algorithmic}
\end{algorithm}

Algorithm~\ref{algo:construct_ei} takes the observer task's worst case execution time $e_o$, the
remaining execution time of the present instance $e_o'$ and the maximum reconstruction
duration $\BUDGET$ as inputs. It outputs the start time $t_{begin}$ and end time
$t_{end}$ of the detected execution interval as well as the updated remaining
execution time of the present instance $e_o'$. \emph{Lines 1 --4} 
initialize the variables to be used by the algorithm. Specifically, \emph{line
3} computes the point in time (the stop condition) when the algorithm reaches
the given maximum reconstruction duration $\BUDGET$ for the present instance.
\emph{Lines 5 -- 9} are used to detect a preemption and check if
current time exceeds the computed stop time point. These lines keep track of the
time difference between each loop by reading present time from a global timer (\ie a system timer)
and comparing it to the time from the previous loop. If the time difference
exceeds what we anticipate (the execution time of the loop), we know 
that a preemption occurred (\ie one or more higher-priority tasks executed). 
The loop exits either when a preemption is
detected or the present time exceeds the computed stop time point. \emph{Lines
10 -- 12} determine the end time of the reconstructing execution
interval. If the loop exits because of a preemption, the last time point
before the preemption is taken as the end time of that execution interval
(\emph{line 11}). Otherwise, no preemption is detected, all $\BUDGET$ duration
is used up and the latest time point is taken as the end time of the execution
interval (\emph{line 13}). \emph{Line 15} updates the remaining execution time
of the present job for the next invocation. \emph{Line 16} returns the
reconstructed execution interval (its start time $t_{start}$ and end time
$t_{end}$) and the updated remaining execution time.

\subsection{Schedule on A Schedule Ladder Diagram}
\label{sec::schedule_on_ladder}
To better understand the effectiveness of the schedule ladder diagram in profiling the victim task's behavior, we plot the original schedule of Example~\ref{example:analyze_ei} on the ladder diagram in Figure~\ref{fig:schedule_on_ladder_example} so that readers get a better sense of it. 
This is not a part of our algorithms, but it gives us an insight into the correlation of the behaviors between the observer task and the victim task. 
%

\begin{figure}[h]
\vspace{-0.5\baselineskip}
\centering
\includegraphics[width=0.65\columnwidth]{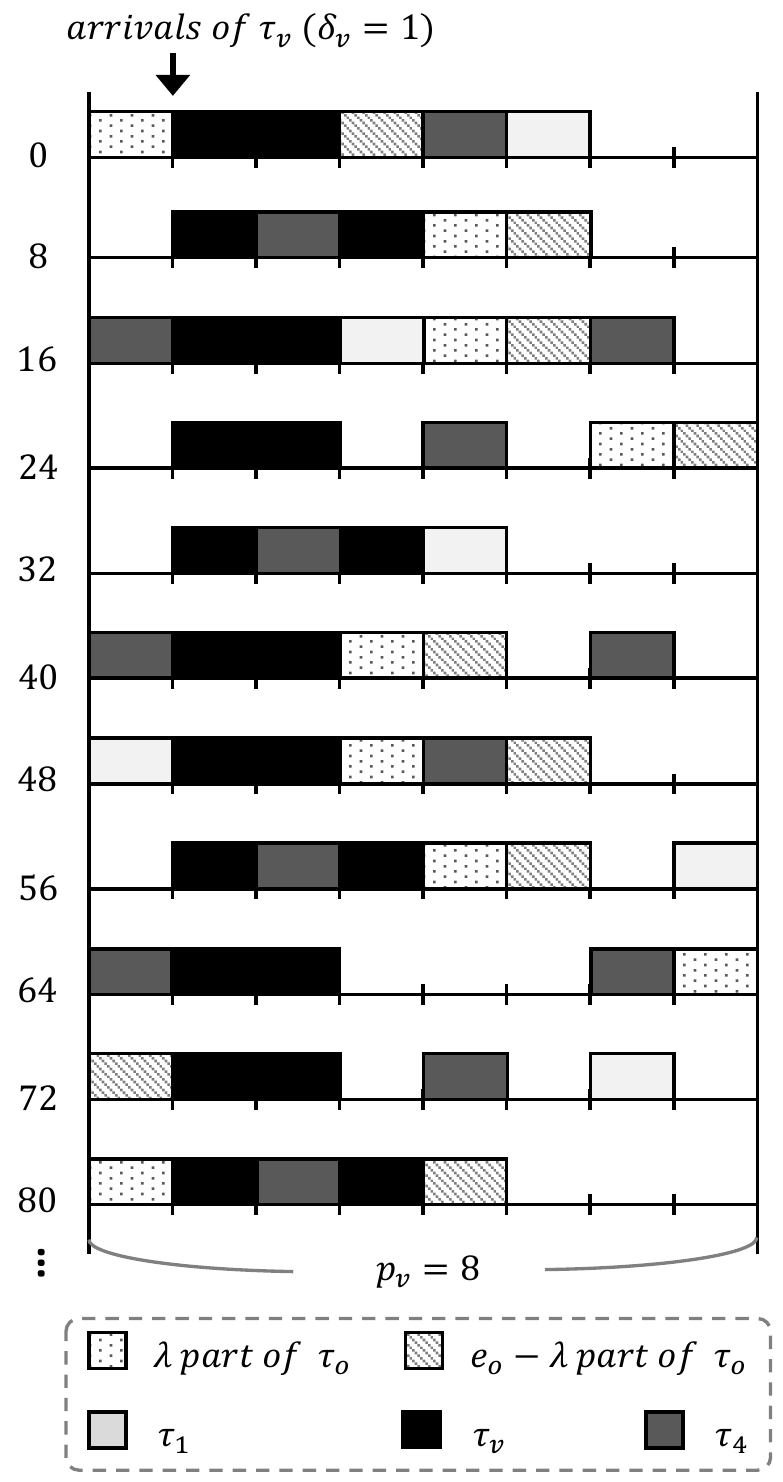}
\caption{The schedule of the task set in Example~\ref{example:analyze_ei} plotted on a schedule ladder diagram with a width of $p_v$. 
It shows that time columns $[1,3)$ are always occupied by either the victim task or other higher priority tasks. Therefore, the execution intervals of the observer task will not land on these time columns where the true arrival column is enclosed. This fact is what the proposed algorithms is based on.}
\label{fig:schedule_on_ladder_example}
\vspace{-0.5\baselineskip}
\end{figure}

\section*{Acknowledgments}
The authors would like to thank the anonymous reviewers and the shepherd for their valuable comments and suggestions. The authors would also like to thank Jesse Walker and Yeongjin Jang for their feedback on earlier versions of the paper. This work is supported by the National Science Foundation (NSF) under grant SaTC-1718952.
Any opinions, findings and conclusions or recommendations expressed in this publication are those of the authors and do not necessarily reflect the views of the NSF.

\bibliographystyle{IEEEtran}
\bibliography{main}

\end{document}